\documentclass[11pt]{article}

\usepackage[latin1]{inputenc}
\usepackage{amsmath,amsfonts,hyperref,amsthm, enumitem, graphicx, comment}
\usepackage[textsize=footnotesize,color=green!40]{todonotes}
\usepackage{fullpage,bbm}
\usepackage{authblk}

\newtheorem{thm}{Theorem}
\newtheorem{lemma}[thm]{Lemma}
\newtheorem{prop}[thm]{Proposition}

\newtheorem{defn}[thm]{Definition}
\theoremstyle{definition}
\newtheorem*{definition*}{Definition}

\newtheorem{cla}[thm]{Claim}

\DeclareMathOperator{\Var}{Var}
\DeclareMathOperator{\Gap}{Gap}

\newcommand{\eps}{\epsilon}

\newcommand{\tmix}{t_{\text{mix}}}

\newcommand{\cB}{\mathcal{B}}

\newcommand{\cE}{\mathcal{E}}
\newcommand{\cF}{\mathcal{F}}

\newcommand{\cL}{\mathcal{L}}

\newcommand{\cR}{\mathcal{R}}

\newcommand{\cV}{\mathcal{V}}

\newcommand{\sm}{\setminus}

    \makeatletter
    \let\@fnsymbol\@arabic
    \makeatother

\title{The Glauber dynamics for edge-colourings of trees}

\author[$1$]{Michelle Delcourt%
}
\author[$2$]{Marc Heinrich%
}
\author[$3$]{Guillem Perarnau%
}
\affil[$1$]{\small\it Department of Mathematics, Ryerson University, Toronto, Canada.}
\affil[$2$]{\small\it LIRIS, Université Claude Bernard, Lyon, France.}
\affil[$3$]{\small\it Departament de Matem\`atiques (MAT), Universitat Polit\`ecnica de Catalunya (UPC), Barcelona, Spain.}

\begin{document}
	
\maketitle
\begin{abstract}
Let $T$ be a tree on $n$ vertices and with maximum degree $\Delta$. We show that for $k\geq \Delta+1$ the Glauber dynamics for $k$-edge-colourings of $T$ mixes in polynomial time in $n$. The bound on the number of colours is best possible as the chain is not even ergodic for $k \leq  \Delta$. Our proof uses a recursive decomposition of the tree into subtrees; we bound the relaxation time of the original tree in terms of the relaxation time of its subtrees using block dynamics and chain comparison techniques. Of independent interest, we also introduce a monotonicity result for Glauber dynamics that simplifies our proof.
\end{abstract}

The Glauber dynamics is a Markov chain over the set of configuration of spin systems in the vertices of a graph. In this paper we consider the discrete-time Metropolis Glauber dynamics and study the particular case of sampling $k$-colourings of graphs. The Glauber dynamics has attracted interest from many different areas. In statistical physics for example, this dynamics gives a single-update sampler for the Gibbs distribution, with the particular case of $k$-colourings corresponding to the antiferromagnetic Potts model at zero temperature. In computer science, rapid convergence of the Glauber dynamics gives a fully polynomial-time randomised approximation scheme for the number of $k$-colourings of a graph, a counting problem that is $\sharp P$-complete. Finally, this dynamics is a popular Monte Carlo method for simulating physical systems as it is both relatively simple and easy to implement. The efficiency of these algorithms depends on the time the chain takes to reach equilibrium, known as the mixing time. This is the main object of study in this paper.

A well-known conjecture in this area states that the Glauber dynamics for $k$-colourings of graphs with $n$ vertices and maximum degree $\Delta$ mixes in polynomial time in $n$ for every $k\geq \Delta+2$. This bound on the number of colours is best possible as the chain might not be ergodic for $k= \Delta+1$. Jerrum~\cite{Jer}, and independently, Salas and Sokal~\cite{SalSok}, proved that polynomial mixing happens for $k\geq 2\Delta$. Vigoda~\cite{Vig} improved the bound to $k\geq 11\Delta/6$ analysing the flip chain. The flip chain is the Wang-Swendsen-Kotecky (WSK) algorithm~\cite{WanSweKot} on $k$-colourings and its transitions are colour-swaps of maximal bicoloured connected subgraphs, instead of single-site recolourings as the Glauber dynamics. Recently, Chen and Moitra~\cite{CheMoi}, and independently, Delcourt, Perarnau and Postle~\cite{DPP} showed that polynomial mixing holds for $k\geq (11/6-\eps)\Delta$, for some small $\eps>0$. This is the best known bound for general graphs. The conjecture has also been studied for several classes of graphs such as graphs with large girth~\cite{DFHV,HayVig}, random graphs~\cite{EHSV,MosSly}, planar graphs~\cite{HVV} and graphs with bounded tree-width~\cite{Var}.

One of the central topics in the area is the study of Glauber dynamics on trees. Although it is quite simple to sample a uniformly random $k$-colouring of a tree, Glauber dynamics on trees often shows an extremal behaviour that helps in the understanding of its performance on general graphs.
First of all, let us remark that the Glauber dynamics on colourings of trees is ergodic if $k\geq 3$. In this context, two important thresholds appear: the \emph{uniqueness threshold} at $k=\Delta+2$~\cite{Jon} that corresponds to the  existence of an infinite-volume Gibbs distribution under any boundary condition and the \emph{reconstruction threshold} at $k=(1+o(1))\Delta/\log{\Delta}$~\cite{Sly}, which marks the existence of an infinite Gibbs distribution under free boundary conditions and has connections to algorithmic barriers for local algorithms on trees and sparse random graphs~\cite{AchCoj}.

Most of the existing results study the complete $b$-ary tree, (where $\Delta = b+1$). Martinelli, Sinclair and Weitz~\cite{MSW} proved that the Glauber dynamics for $k$-colourings on $b$-ary trees with $n$ vertices mixes in time $O(n\log{n})$ for every $k\geq b+2$ and for any boundary condition. Hayes, Vera and Vigoda~\cite{HVV} showed that this dynamic mixes in polynomial time (even for planar graphs) for $k\geq C\Delta/\log{\Delta}$ for some constant $C$ and it might take superpolynomial time for $k=o(\Delta/\log{\Delta})$, provided that $\Delta,k\to\infty$ as 
$n\to\infty$. For complete trees, Tetali, Vera, Vigoda, and Yang~\cite{TVVY} established that the threshold for polynomial mixing coincides with the reconstruction threshold, and Sly and Zhang~\cite{SlyZha}  extended the approach in~\cite{MSW} to show $O(n \log n)$ mixing in the non-reconstruction regime.
The mixing time has also been studied beyond the reconstruction threshold. Berger, Kenyon, Mossel, and Peres~\cite{BKMP} proved that the mixing time in a complete tree of maximum degree $\Delta$ is polynomial for $k\geq 3$ and constant~$\Delta$. Lucier, Molloy, and Peres~\cite{LMP,LucMol} obtained the explicit polynomial bound $n^{O(1+\Delta/k\log{\Delta})}$ for every $k\geq 3$. This result extends to any tree with maximum degree $\Delta$ and it is best possible up to the constant in the exponent~\cite{GolJerKar,LucMol}.

In this paper we initiate the study of Glauber dynamics for edge-colourings of graphs by bounding the mixing time for the dynamics on edge-colourings of trees with maximum degree $\Delta$. To our knowledge, there are no previous results specific to edge-colourings, although the results for vertex-colourings of general graphs can be transferred as the set of edge-colourings of a graph corresponds to the set of vertex-colourings of its line graph. Note that if a graph has maximum degree $\Delta$, then its line graph has maximum degree at most $2\Delta-2$. Thus, the results in~\cite{CheMoi, DPP} imply that the Glauber dynamics on edge-colourings on any graph with maximum degree $\Delta$ mixes in polynomial time for every $k\geq (11/3-\eps)\Delta$.

Edge-colourings behave substantially different than vertex-colourings. Vizing's theorem states that an edge-colouring exists if $k\geq \Delta+1$, although the line graph can have maximum degree $2\Delta-2$. Therefore, the Glauber dynamics for $k$-edge-colourings is ergodic for every $k\geq 2\Delta$, and this is best possible as the Knesser graph of degree $\Delta$, admits a frozen edge-colourings with $2\Delta-1$ colours~\cite{Knesser}.

The analysis for trees is simpler. The edges of a tree with maximum degree $\Delta$ can be coloured with $k=\Delta$ colours and no $k$-colouring exists for $k< \Delta$. However, at $k=\Delta$, the Glauber dynamics is not ergodic for \emph{any} tree with maximum degree $\Delta$, as the edges incident to a vertex of degree $\Delta$ cannot be recoloured. Even more, 
on a complete tree, any (non-trivial) block dynamics on its edge-colourings with $\Delta$ colours is not ergodic, as edges in the boundary of a block cannot be updated to any colour using block dynamics transitions. 
It is not difficult to see that the Glauber dynamics on edge-colourings of a tree with maximum degree $\Delta$ is ergodic if $k\geq \Delta+1$. Note that this bound on $k$ is roughly half the bound for ergodicity of vertex-colourings in general graphs, as the maximum degree of the line graph is at most $2\Delta-2$.

In sharp contrast to the vertex-colourings case, this paper shows that edge-colouring dynamics on trees has polynomial mixing as soon as it becomes ergodic.
\begin{thm}\label{thm:stronger}
There exists $C>0$ such that for any tree $T$ on $n$ vertices and maximum degree $\Delta\geq 3$ and for any $k\geq \Delta+1$, the Glauber dynamics for $k$-edge-colourings on $T$ has mixing time at most $n^{C}$.
\end{thm}

We now give a brief overview of the strategy used to prove Theorem~\ref{thm:stronger}. 
First of all, we will work with the continuous-time version of the Glauber dynamics. Although the proofs can be adapted to the discrete-time setting, we use the continuous-time setting, mirroring previous work in~\cite{LucMol,LMP}.
The bound on the mixing time of the dynamics will follow from bounding its relaxation time. We first show that it suffices to bound the relaxation time of $d$-regular trees where $d=k-1$, which notably simplifies the proof. To this end, we prove a monotonicity result on the relaxation time of the Glauber dynamics that can be of independent interest (see Section~\ref{sec:mono}).

The main idea of the proof is to recursively decompose the tree into subtrees, and study the process restricted to each one. The decomposition approach has already been used in the literature to bound the mixing time of the Glauber dynamics for vertex-colourings of trees~\cite{BKMP,LucMol,LMP}. To analyse the decomposition procedure we need to study the associated block dynamics, which can be informally described as follows: given a partition of the tree into subtrees, at each step we select one subtree and update its colouring by choosing a uniformly random colouring that is compatible with the boundary condition. We then use a result of Martinelli~\cite{Mar} on block dynamics to upper bound the relaxation time of the whole process in terms of the relaxation time of each block, and the relaxation time of the block dynamics. 

Our decomposition procedure will satisfy two properties,
\begin{itemize}
	\item[i)] the Glauber dynamics on each subtree is ergodic for every possible boundary condition;
	\item[ii)] the number of recursive decompositions is small (i.e., logarithmic in the size of the tree).
\end{itemize}
Condition i) is necessary to apply Martinelli's result on block dynamics.
Moreover, if ii) holds, the upper bound we obtain on the relaxation time is polynomial in $n$, with an exponent that is independent from $\Delta$ and $k$.

The splitting strategy builds on the one used in~\cite{LucMol,LMP}, and introduces new ideas to deal with the extremal case $k=\Delta+1$ (see Section~\ref{sec:proofThm} for a precise description). Informally speaking, to ensure that ii) holds, at every step we split into several trees of at most half the size of the original one. To ensure that i) holds, we make sure that every subtree has at most two edges incident to the boundary, and that these two edges are \emph{non-adjacent}. The latter condition is not necessary for $k\geq \Delta+2$ but it is crucial for $k=\Delta+1$, as if the boundary edges are adjacent, the Glauber dynamics on each block might be reducible. It is also novel with respect to the approach in~\cite{LMP} where a different chain was introduced to deal with the irreducibility of the blocks. The strategy is implemented by splitting the tree so that all subtrees pend from a vertex, or an edge. 

The second ingredient of the proof is a bound on the relaxation time of the block dynamics, where each block corresponds to a subtree hanging from a root vertex. This dynamics is very similar to the Glauber dynamics on the star spanned by the  edges incident to the root, constrained to a boundary condition.
In Section~\ref{sec:relaxBlockTree} we reduce the block dynamics to an analogous problem: bounding the relaxation time of list-vertex-colourings of a clique (the line graph of a star), where the lists are given by the boundary constraints and each vertex is updated at a different rate. In Sections~\ref{sec:clique_1} and~\ref{sec:clique_2}, we use the weighted canonical paths method of Lucier and Molloy~\cite{LucMol} and a multi-commodity flow argument to bound the relaxation time of the list-colourings of the clique. 

To conclude the proof in the case $k=\Delta+1$, we need to obtain a bound on the relaxation time of the block dynamics where the blocks correspond to trees hanging from a root edge instead of a vertex. In this case, it suffices to study the list-vertex-colouring dynamics on the graph formed by two cliques intersecting at a vertex, which we do in Section~\ref{sec:biclique}.

\section{Preliminaries}

Let $G=(V,E)$ be a finite graph. The \emph{line graph} $G^\ell$ of $G$ is the graph with vertex set $E$ where two edges are adjacent in $G^\ell$ if and only if they are incident in $G$. We define some notions for vertex-colourings of graphs, but these can be naturally transferred to the setting of edge-colourings by considering line graphs.

For $v\in V$, the \emph{neighbourhood of $v$}, denoted by $N(v)$, is the set of vertices in $G$ adjacent to $v$. For $k\in \mathbb{N}$, denote $[k]=\{1,\dots, k\}$. A \emph{(proper) $k$-vertex-colouring} is a function $\mu : V\to [k]$ such that $\mu(v)\neq \mu(u)$ for every $u\in N(v)$. All colourings in this paper are proper. 
Let $\Omega_{V}$ denote the set of $k$-vertex-colourings of $G$. The set of $k$-edge-colourings (i.e., $k$-vertex-colourings of $G^\ell$) will be denoted by $\Omega_{E}$. Throughout the paper, we use similar notation to distinguish the vertex and the edge-version of each set or parameter.

For $\mu\in \Omega_{V}$ and $U\subseteq V$, we denote the restriction of $\mu$ to $U$ by $\mu|_U$. We denote by $\mu(U)$ the set of colours in $U$. 
We write $\Omega^\mu_U$ to denote the set of $\sigma\in \Omega_V$ which agree with $\mu$ on $V\setminus U$, i.e., with $\sigma|_{V\setminus U}=\mu|_{V\setminus U}$. Informally, we think of $\Omega^\mu_U$ as colourings of $U$ which are compatible with $\mu$ in the boundary of $U$. 

A list assignment of $V$ is a function $L: V\to 2^{[k]}$. An \emph{$L$-colouring of $G$} is a $k$-colouring $\mu\in \Omega_V$ satisfying $\mu(v)\in L(v)$ for all $v\in V$. We denote by $\Omega^L_V$ the set of all $L$-colourings of $G$. Note that if $U\subseteq V$ and $H$ is the subgraph of $G$ induced by $U$, then any $\mu\in \Omega_V$ yields a list assignment for $U$ where $L^{\mu}(u)=[k]\setminus \mu(N(u) \sm U)$. This gives a natural bijection between $\Omega^\mu_U$ and $\Omega^{L^\mu}_U$.

In this paper we will focus on $k$-edge-colourings of a tree $G$ on $n$ vertices with maximum degree at most~$\Delta$. Note that $G^{\ell}$ is a union of cliques of size at most $\Delta$ such that two cliques intersect in at most one vertex and every cycle is contained in some clique.

The vertices of $G$ with degree $1$ are called \emph{leaves} and the vertices with degree at least $2$, \emph{internal} vertices.  Throughout this paper we will denote the set of edges in $G$ incident to an edge $e$ by $N(e)$. 
If $T$ is a subtree of $G$, we define the \emph{exterior and interior (edge) boundary of $T$} respectively as 
\begin{align*}
\partial_{G} T & = \{e\in E\setminus E(T): N(e)\cap E(T)\neq \emptyset\}\;,\\
\overline{\partial}_{G} T & = \{e \in E(T): N(e)\cap \partial_{G} T\neq \emptyset\}\;.
\end{align*} 
  If $G$ is clear from the context, we will denote them by $\partial T$ and $\overline\partial T$. We said that $T$ has a \emph{fringe boundary} if all edges in $\overline\partial T$ have an endpoint that was a leaf of $T$. We will use $t$ to denote the size of $\overline\partial T$.  In this paper, we will always consider $t$ to be a constant with respect to $n, \Delta, k$. 

\subsection{Markov chains and Glauber dynamics}

Let $X_t$ be a discrete-time Markov chain with state space $\Omega$ with transition matrix $P$.
It is often convenient to consider a continuous-time analogue of $(X_t)$. For that, consider the continuous-time Markov chain with state space $\Omega$ and entries of its generator matrix $\cL$ given as follows: non-diagonal entries are the corresponding entries of $P$ and diagonal entries are such that each row of $\cL$ adds to $0$.  In a slight abuse of notation, we denote the chain by $\cL$.  Here $\cL[x \to y]$ represents the rate at which $x$ jumps to $y \neq x$. For a dynamics $\cL$ we write $(x,y)\in \cL$ if and only if $\cL[x\to y]>0$. 

Given probability distributions $\nu,\pi$ on $\Omega$, we define their \emph{total variation distance} as
$$
d_{TV} (\nu,\pi) = \frac{1}{2} \sum_{x\in \Omega} |\nu(x)-\pi(x)|\;.
$$
Let $\cL$ be an ergodic Markov chain with stationary distribution $\pi$.
The \emph{mixing time} of $\cL$ is defined as  
$$
\tmix(\cL) = \inf\left\{t: \max_{x\in \Omega}  d_{TV} (\nu^t_x,\pi)<1/4\right\}\;,
$$
where $\nu^t_x$ is the probability distribution on $\Omega$ obtained by running $\cL$ for time $t$ starting at $x$.
The \emph{spectral gap} of $\cL$, denoted by $\Gap(\cL)$, is the second smallest eigenvalue of $-\cL$. The \emph{relaxation time} of $\cL$ is defined as
$$
\tau(\cL)=\frac{1}{\Gap(\cL)}\;.
$$
If the chain is irreducible, it satisfies
\begin{align}\label{eq:mixrel}
\tmix(\cL) \leq  \log (4/\pi_{\text{min}}) \tau(\cL)\;.
\end{align}
where  $\pi_{\text{min}}= \min_{x\in \Omega} \pi(x)$ (see, e.g. Theorem 12.3 in~\cite{LevPer}).

For a graph $G=(V,E)$ with $V= \{v_1,\dots,v_n\}$, a probability distribution $\nu$ on $V$, and a positive integer $k$, the
 \emph{(discrete-time)} Glauber dynamics for $k$-vertex-colourings of $G$ with distribution $\nu$ is a discrete-time Markov chain $X_t$ on $\Omega_V$, where $X_{t+1}$ is obtained from $X_t$ by choosing a $v\in V$ according to $\nu$ and $c\in [k]$ uniformly at random, and updating $v$ with $c$ if this colour does not appear in $N(v)$. If $\nu$ is uniform, then we call $X_t$ the (discrete-time) Glauber dynamics for $k$-vertex-colourings of $G$.

The \emph{continuous-time Glauber dynamics for $k$-vertex-colourings of $G$ with parameters $(p_1,\dots, p_n)$} is a continuous-time Markov chain on $\Omega_V$ with generator matrix given by
\begin{align*}
\cL_V[\sigma\to \eta]&=\begin{cases}
p_i & \text{if }\sigma,\eta\text{ differ only at $v_i$,}\\
0& \text{otherwise.}
\end{cases}
\end{align*}
One can imagine this stochastic process as every vertex $v_i$ having an independent rate $k p_i$ Poisson clock. When the clock at $v$ rings, one chooses a colour $c\in [k]$ uniformly at random and recolours $v$ with $c$ if possible. 

If $(p_1,\dots,p_n)$ are a probability distribution, then $\cL_V$ is the continuous-time version of the discrete-time Glauber dynamics with such distribution.
Here we will consider $p_i$ that are bounded away from $0$ as $n\to \infty$, which in particular implies that the continuous chain updates vertices much faster than the discrete version.
In particular, if $p_i=1/k$ for every $i\in [n]$, then we call $\cL_V$ the continuous-time Glauber dynamics for $k$-vertex-colourings of $G$. If all $p_i$ are uniformly bounded, then the expected update rate is $\Theta(n)$ and it follows from standard Markov chain comparison results~(see e.g.~\cite{RanTet}) that
\begin{align}\label{eq:cont_disc}
\tmix(X_t) =O( n\, \tmix(\cL_V) )\;.
\end{align}

As $\cL_V$ is symmetric for any set of parameters (i.e. $\cL_V[\sigma\to \eta]=\cL_V[\eta\to\sigma]$), if $\cL_V$ is ergodic then its stationary distribution $\pi$ is uniform on $\Omega_V$. As $|\Omega_V|\leq k^n$, we have that  $\pi_{\min}\geq k^{-n}$ and by~\eqref{eq:mixrel}, the mixing time of Glauber dynamics satisfies,
\begin{align}\label{eq:mixrel2}
\tmix(\cL_V) \leq  (n\log{k}) \tau(\cL_V)\;.
\end{align}

If $G=(V,E)$ is a tree with maximum degree $\Delta$ its line graph $G^\ell$ is $(\Delta-1)$-degenerate (i.e. every subgraph has a vertex of degree at most $\Delta-1$). It is known~\cite{DFFV} that in this case the Glauber dynamics $\cL_{E}$ for $k$-edge-colourings of $G$ is ergodic provided $k\geq \Delta+1$. 

For $U\subseteq V$ and $\mu\in \Omega_V$, denote by $\cL^\mu_U$ the dynamics defined by $\cL_V$ restricted to the set $\Omega_U^\mu$. Similarly, for a list assignment $L$ of $U$, let $\cL_U^L$ be the dynamics defined on the state space $\Omega^L_U$. All these chains are symmetric but their ergodicity will depend on the size of the boundary of $\mu$ in $U$, and the size of the lists $L(u)$ for $u\in U$. If ergodic, then their stationary distribution will be uniform in the corresponding state space.

\section{Comparison of Markov chains}\label{sec:tools}

This section introduces the methods in the analysis of the relaxation time of reversible Markov chains that we will use later in the proofs. The first result is a unified framework of two well-known chain comparison techniques. We then define the (reduced) block dynamics and present known results on how this can be used to bound the original chain. Finally, we provide a result on the monotonicity of the Glauber dynamics that will allow us to notably simplify the proof.

\subsection{The weighted multi-commodity flow method}\label{sec:comparison}

In this subsection we present a unified framework for two similar techniques for comparing the relaxation time of two different Markov chains $\cL$ and $\cL'$ on the same state space. These two techniques are (1) the \emph{fractional paths} (or \emph{multi-commodity flows}) method~\cite{Sin} and (2) the \emph{weighted canonical paths} method~\cite{LucMol}. Both methods are generalisations of the \emph{canonical paths method} (see e.g.~\cite{DiaSal}). The main idea of this method is to simulate the transitions of $\cL'$ using transitions of $\cL$ in such a way that no transition of $\cL$ is used too often, to obtain an upper bound on the ratio between $\tau(\cL)$ and $\tau(\cL')$.

Consider two continuous-time ergodic reversible Markov chains $\cL$ and $\cL'$ on $\Omega$ with stationary distributions $\pi$ and $\pi'$, respectively.
In the following, we denote by $\omega \colon \Omega \times \Omega \to \mathbb{R}$ a weight function on the transitions of $\cL$.  
To each $\alpha,\beta\in \Omega$ with $(\alpha,\beta)\in \cL'$, we associate a set of paths $\Gamma_{\alpha, \beta}$ where every $\gamma\in \Gamma_{\alpha,\beta}$ is a sequence $\alpha=\xi^0,\dots ,\xi^m=\beta$, for some $m\geq 1$ with $(\xi^{i-1},\xi^i)\in \cL$ for every $i\in [m]$. We define a \emph{flow from $\alpha$ to $\beta$} to be a function $g_{\alpha,\beta}\colon \Gamma_{\alpha, \beta} \to [0,1]$ such that $\sum_{\gamma \in \Gamma_{\alpha, \beta}} g_{\alpha,\beta}(\gamma) = 1$. The \emph{weight} of $\gamma$ with respect to $\omega$ is defined as 
$$
|\gamma|_{\omega}:=\sum_{i=1}^m \omega(\xi^{i-1},\xi^i)\;.
$$
If the weight function $\omega$ is a constant equal to $1$, then the weight of $\gamma$ is simply denoted by $|\gamma|$, and corresponds to the length of the transformation sequence.

Let $b:= \max_{\alpha \in \Omega}\frac{\pi(\alpha)}{\pi'(\alpha)}$. For every $(\sigma,\eta)\in \cL$ we define its \emph{congestion} as
$$
\rho_{\sigma,\eta}:=\frac{1}{\pi(\sigma)\cL[\sigma \rightarrow \eta] \omega(\sigma, \eta)} \sum_{(\alpha,\beta) \in \cL'} \sum_{\gamma\in\Gamma_{\alpha, \beta}\atop \gamma\ni  (\sigma, \eta)} g_{\alpha,\beta}(\gamma) \pi'(\alpha)\cL'[\alpha \rightarrow \beta] |\gamma|_\omega \;. 
$$
Let $\rho_{\max} = \max \{\rho_{\sigma,\eta}: (\sigma,\eta)\in \cL\}$ be the maximum congestion over all transitions of $\cL$.
We are now ready to state our main lemma.
\begin{prop}[Weighted multi-commodity flows method]\label{prop:weighted_frac_paths}
We have 
$
\tau(\cL) \leq b^2 \rho_{\max}\, \tau(\cL') \;.
$
\end{prop}

The proof of this result is based on the same ideas as the proofs of both the weighted canonical paths method and the fractional paths method, providing a common framework for them. It is included in the Appendix for sake of completeness, and follows from standard computations using the variational characterisation of $\Gap(\cL)$ that involves the variance and the Dirichlet form. We will not need this result in full generality in our proofs, but we believe it is interesting in its own right and will use it in two particular cases.

The first case is when the stationary distribution of the two chains is uniform over $\Omega$, and for every pair $(\alpha, \beta) \in \cL'$ the set of paths $\Gamma_{\alpha, \beta}$ consists of a single path $\gamma_{\alpha, \beta}$, with $g_{\alpha,\beta}(\gamma_{\alpha, \beta}) = 1$. In this case, Proposition~\ref{prop:weighted_frac_paths} implies the result in~\cite{LucMol}.
\begin{prop}[Weighted canonical paths method]
If $\pi$ and $\pi'$ are uniform, then
	\label{prop:weighted_cano_path}
$$
\tau(\cL) \leq \tau(\cL') \cdot \max_{(\sigma, \eta) \in \cL} \left( \frac{1}{\cL[\sigma \rightarrow \eta] \omega(\sigma, \eta)} \sum_{(\alpha,\beta) \in \cL' \atop \gamma_{\alpha, \beta} \ni  (\sigma, \eta)}  \cL'[\alpha \rightarrow \beta] |\gamma_{\alpha, \beta}|_\omega \right) \;.
$$	
\end{prop}

In the second case, all transitions have weight $\omega(\sigma, \eta) = 1$, and again both chains have uniform stationary distributions. In this case, we obtain the following.

\begin{prop}[Fractional paths method]\label{prop:frac_path}
If $\pi$ and $\pi'$ are uniform, then
	$$\tau(\cL) \leq \tau(\cL') \cdot \max_{(\sigma, \eta) \in \cL} \left( \frac{1}{\cL[\sigma \rightarrow \eta]} \sum_{(\alpha,\beta) \in \cL'} \sum_{\gamma\in\Gamma_{\alpha, \beta}\atop \gamma\ni  (\sigma, \eta)} g_{\alpha,\beta}(\gamma) \cL'[\alpha \rightarrow \beta] |\gamma| \right)\;.$$
\end{prop}

Finally, we will use the following well-known fact: if two Markov chains have the same transitions with similar transition rates, and similar stationary distributions, then their mixing time is also similar (see~\cite[Theorem~3]{RanTet}; it also follows as a corollary of Proposition~\ref{prop:weighted_frac_paths}).
\begin{prop}\label{prop:constant_dyn}
	Suppose there exists a constant $c > 1$ such that for every $\alpha, \beta \in \Omega$ it holds that
	\begin{itemize}
		\item[-] $\frac 1 c \pi'(\alpha) \leq \pi(\alpha) \leq c \pi'(\alpha)$;
		\item[-] $\frac 1 c \cL'[\alpha \rightarrow \beta] \leq \cL[\alpha \rightarrow \beta] \leq c \cL'[\alpha \rightarrow \beta]\;.$
	\end{itemize}
	Then, there is a constant $K>0$ such that
	$$\frac 1 K \tau(\cL') \leq \tau(\cL) \leq K \tau(\cL') \;.$$
\end{prop}

\subsection{Weighted and Reduced Block Dynamics}\label{sec:block}

This subsection describes block dynamics and its reduced version. Informally speaking, block dynamics is a generalisation of Glauber dynamics where one splits the set of vertices into ``blocks'' (usually with few intersections/interactions between them), and updates each block at a time, according to some probabilities (see~\cite{Mar}). In this paper we will only consider disjoint blocks that partition the vertex set. Note that if all blocks are singletons, we recover Glauber dynamics. The method is presented for vertex-colourings, but it works the same for edge-colourings by taking an edge partition instead. We will present some known results for bounding the relaxation time of Glauber dynamics in terms of the relaxation time of its (reduced) block dynamics. 

Let $G=(V_G,E_G)$ be a graph, $T=(V,E)$ be an induced subgraph and let $\cV=\{V_1,\dots, V_r\}$ be a \emph{partition} of $V$. Let $\mu \in \Omega_{V_G}$.
Suppose $\cL^\sigma_{V_i}$ is ergodic for every $i\in [r]$ and every $\sigma\in \Omega^{\mu}_{V}$ and let $\pi_{V_i}^\sigma$ denote its stationary distribution, which is also the uniform distribution on $\Omega_{V_i}^\sigma$. The \emph{weighted block dynamics on $\cV$ with boundary condition $\mu$} is a continuous-time Markov chain with state space $\Omega^\mu_{V}$ and generator matrix $\cB^\mu_{\cV}$ given for any $\sigma\neq \eta$ by
\begin{align*}
\cB^\mu_{\cV}[\sigma \rightarrow \eta]&=\begin{cases}
g_i \pi_{V_i}^\sigma(\eta) & \text{if }\eta\in \Omega^{\sigma}_{V_i },\\
0& \text{otherwise,}
\end{cases}
\end{align*}
where $g_i = \min_{\sigma \in \Omega^\mu_{V}} \Gap(\cL^\sigma_{V_i})$ is the minimum gap for the Glauber dynamics on the block $V_i$, where the minimum is taken over all possible boundary conditions.
Note that $\cB^\mu_{\cV}$ can be understood as the dynamics where each block $V_i$ updates its entire colouring at times given by an independent Poisson clock of rate $g_i$. The new colouring of $V_i$ is chosen uniformly among all the possible colourings which are compatible with the current boundary conditions. It is clear that the block dynamics is ergodic if the Glauber dynamics is, since each transition of the Glauber dynamics is a valid transition of the block dynamics. Moreover, both dynamics have the same stationary distribution.
 
An unweighted version of the block dynamics, corresponding to $g_i=1$, was used by Martinelli in~\cite{Mar}. 
Lucier and Molloy generalised this result for weighted block dynamics:
\begin{prop}[Proposition 3.2 in~\cite{LucMol}]\label{prop:mar}
For every $\mu\in \Omega_{V_G}$ and partition $\cV$ of $V$, we have
$$
\tau(\cL^\mu_{V})\leq \tau(\cB^\mu_{\cV})\;.
$$
\end{prop}

Given a block partition $\cV$, let $H=(U,F)$ be the subgraph of $T$ composed of the vertices adjacent to vertices in other blocks. Let $\Omega_{\cR}$ be the set of colourings of $U$ induced by the colourings in $\Omega^\mu_{V}$; we will use $\hat \sigma,\hat \eta, \dots$ to denote the elements of $\Omega_\cR$. It is convenient to see $\Omega_\cR$ as a set of colouring classes of $\Omega_V^\mu$ where two colourings are equivalent if and only if they coincide on $U$. More precisely, for $\hat \sigma\in \Omega_\cR$, let $\Omega^{\hat \sigma}_*$ be the set of colourings $\sigma\in \Omega_V^\mu$ with $\sigma|_U=\hat \sigma$. In a slight abuse of notation, we write $\Omega_{V_i}^{\hat \sigma}$ to denote the set $\{\eta\in \Omega_{V_i}^{\sigma}:\,\sigma\in \Omega^{\hat\sigma}_*\}$, and we write $\pi_{V_i}^{\hat \sigma}$ for $\pi_{V_i}^{\sigma}$ where $\sigma$ is an arbitrary colouring in $\Omega_{*}^{\hat \sigma}$. Note that the projection of $\pi_{V_i}^{\hat \sigma}$ onto $U$ is well-defined and independent of the choice of $\sigma\in \Omega_*^{\hat\sigma}$.

The \emph{reduced} version of $\cB_{\cV}^\mu$, is a dynamics with state space $\Omega_{\cR}$ and generator matrix $\cR_{\cV}^\mu$ given for any $\hat \sigma\neq \hat \eta$ by
\begin{align}
\cR_{\cV}^\mu[\hat \sigma \rightarrow \hat \eta]&=\begin{cases}
g_i \pi^{i, \hat \sigma}_{\text{proj}}(\hat \eta)  
& \text{if }\hat \eta=\eta|_U \text{ for some }\eta\in \Omega^{\hat \sigma}_{V_i },\\
0& \text{otherwise.}
\end{cases}
\end{align}
where $\pi^{i, \hat \sigma}_{\text{proj}}$ is the projection of $\pi_{V_i}^{\hat\sigma}$ onto $U$; that is, for an arbitrary $\sigma\in \Omega_*^{\hat\sigma}$, the probability distribution on the colourings of $U$ defined for any $\hat \eta \in \Omega_\cR$ by
\begin{align}\label{eq:proj_2}
	\pi^{i, \hat \sigma}_{\text{proj}}(\hat \eta) = \pi_{V_i}^\sigma(\{\eta \in \Omega_{V_i}^\sigma,\, \eta|_{U} = \hat \eta \}) \; .
\end{align}

In the particular case where each block contains only one vertex in $H$, only one vertex of $H$ changes colour during a transition of $\cR_\cV^\mu$. In this case, the reduced block dynamic is very similar to the Glauber dynamics on $H$ with some parameters $p_i$ determined by the Glauber dynamics on $V_i$, only with slightly different transition rates.

We will use another result of Lucier and Molloy that shows that the weighted block dynamics and the reduced block dynamics have the same relaxation time.
\begin{prop}[Proposition 3.3 in~\cite{LucMol}]\label{prop:red}
For every $\mu\in \Omega_{V_G}$ and partition $\cV$ of $V$, we have
	$$
	\tau(\cB^\mu_{\cV}) = \tau(\cR^\mu_{\cV})\;.
	$$
\end{prop}

Finally, we will use the following property. If the reduced block dynamics is ergodic, then the projection of $\pi_{V}^\mu$ onto $U$ is its stationary distribution.

\begin{lemma}
	\label{lem:reversible_red}
	The reduced block dynamics $\cR_{\cV}^\mu$ is reversible for the projection of $\pi_{V}^\mu$ onto $U$.
\end{lemma}

\begin{proof}

Recall that $H=(U,F)$ is the subgraph of $T$ induced by the vertices adjacent to vertices in other blocks.
	Let $\pi_{\text{proj}}$ be projection of  $\pi_V^\mu$ onto $U$; that is, the probability distribution on the colourings of $U$ defined for any $\hat \sigma \in \Omega_\cR$ by
\begin{align}\label{eq:proj}
	\pi_{\text{proj}}(\hat \sigma) = \pi_V^\mu(\{\sigma \in \Omega_V^\mu,\, \sigma|_{U} = \hat \sigma \}) \; .
\end{align}
	First observe that, since $H$ consists of all the vertices in each block which are adjacent to vertices of other blocks, for any $\hat \sigma \in \Omega_\cR$, any extension of $\hat \sigma$ to $V$ is obtained by computing an extension on each of the blocks separately. As $\pi_V^\mu$ is uniform, we have
	\begin{align*}
	\pi_{\text{proj}}(\hat \sigma) &= \frac {|\{\sigma \in \Omega_{V}^\mu,\, \sigma|_{U} = \hat \sigma  \}|} {|\Omega_{V}^\mu|} 
	= \frac {\prod_{j = 1}^r |\{\sigma \in \Omega_{V_j}^{\hat \sigma},\, \sigma|_{U} = \hat \sigma  \}| } {|\Omega_{V}^\mu|} \;.
	\end{align*}
	Thus, it follows that, for any two colourings $\hat\sigma, \hat \eta \in \Omega_\cR$ which differ only on the block $V_i$, we have
	\begin{align*}
	\pi_{\text{proj}}(\hat \sigma) \cR_\cV^\mu[\hat\sigma \rightarrow \hat\eta] &= \frac {\prod_{j = 1}^r |\{\sigma \in \Omega_{V_j}^{\hat\sigma},\, \sigma|_{U} = \hat\sigma  \}| } {|\Omega_{V}^\mu|} \cdot g_i \frac {|\{\eta \in \Omega_{V_i}^{\hat\sigma}, \ \eta|_{U} = \hat\eta \}| } { |\Omega_{V_i}^{\hat\sigma}|} \\ 
	&= \frac {g_i \cdot |\{\sigma \in \Omega_{V_i}^{\hat\sigma},\, \sigma|_{U} = \hat \sigma  \}| \cdot |\{\eta \in \Omega_{V_i}^{\hat\sigma},\, \eta|_{U} = \hat\eta  \}|}{|\Omega_{V}^\mu| |\Omega_{V_i}^{\hat\sigma}|} \cdot \prod_{j \neq i} |\{\sigma \in \Omega_{V_j}^{\hat \sigma},\, \sigma|_{U} = \hat\sigma  \}|\;.
	\end{align*}
	This quantity is symmetric in $\hat \sigma$ and $\hat \eta$. Indeed,  since $\hat \sigma$ and $\hat\eta$ only differ on $V_i$, we have $\Omega_{V_i}^{\hat\sigma} = \Omega_{V_i}^{\hat\eta}$. Moreover, for every $j \neq i$, $\hat\sigma$ and $\hat\eta$ agree on $V_j \cap U$. As a consequence, we have
		$$ |\{\sigma \in \Omega_{V_j}^{\hat\sigma},\, \sigma|_{U} = \hat\sigma  \}| = |\{\eta \in \Omega_{V_j}^{\hat\eta},\, \eta|_{U} = \hat \eta  \}|\;.$$

	Thus, it follows that $\pi_{\text{proj}}(\hat\sigma) \cR_\cV^\mu[\hat\sigma \rightarrow \hat\eta] = \pi_{\text{proj}}(\hat\eta) \cR_\cV^\mu[\hat\eta \rightarrow \hat\sigma]$, and  $\cR_{\cV}^\mu$ is reversible for $\pi_{\text{proj}}$.
\end{proof}

\subsection{Monotonicity of Glauber dynamics}\label{sec:mono}

Finally, we introduce a monotonicity statement that will allow us to simplify some of our proofs. The previous subsections gave tools to compare the relaxation time of two Markov chains with the same state space but different transitions. Here we are interested in comparing Markov chains with similar transitions but different state spaces. A natural example is comparing the relaxation time of the Glauber dynamics on $G$ and $H$, where $H$ is a subgraph of $G$. In general, it is not clear which of the two relaxation times should be smaller, however, if $H$ and $G$ have a particular structure, we will be able to derive a monotonicity bound.

\begin{prop}\label{prop:mono}
Let $G=(V,E)$ be a graph on $n$ vertices, and $k$ be a positive integer. Let $v\in V$ such that $N(v)$ induces a clique of size at most $k-2$. For any choice of parameters $(p_1, \ldots, p_n)$, the Glauber dynamics $\cL_V$ and $\cL_{V\setminus \{v\}}$ for $k$-colourings of $V$ and $V\setminus\{v\}$ respectively and defined with the same parameters satisfy,
	$$
	\tau\left(\cL_{V \setminus \{v\}} \right) \leq \tau(\cL_V) \;.
	$$
\end{prop}
The proof follows from standard computations and is included in the Appendix. Our proof strongly uses the fact that the neighbourhood of $v$ is a clique. One can easily generalise it to obtain the following bound: if $N(v)$ induces a clique of size $d$ minus $f$ edges, then 
$$
\tau\left(\cL_{V \setminus \{v\}} \right) \leq \bigg(1+\frac{f}{k-d}\bigg)\tau(\cL_V).
$$

\section{Glauber dynamics for list-colourings of a clique}\label{sec:clique}

In this section we analyse the relaxation time of the Glauber dynamics for list-vertex-colourings of a clique. Using the line graph, vertex-colourings of cliques correspond to edges-colourings of the edges incident to a given vertex. Thus, the results in this section will be used later in the proof of our main result for edge-colourings of trees.

Consider the clique with vertex set $U=\{u_1,\ldots, u_d\}$ with $d\geq 1$. Throughout this section we fix the number of colours to $k=d+1$. Recall that given a list assignment of $U$, we can define $\Omega^L_U$ and the dynamics $\cL^L_U$, governed by the parameters $(p_1,\dots, p_d)$, where $p_i$ is the rate at which vertex $u_i$ changes to $c\in L(u_i)$.

For a positive integer $t$, a list assignment $L$ is \emph{$t$-feasible} if it satisfies
\begin{itemize}
\item[-] $|L(u_i)|\geq t+1$ for every $i\in [t]$;
\item[-] $|L(u_i)|=d+1$ for every $i\in [d]\setminus [t]$.
\end{itemize}
We say $u\in U$ is \emph{free} if $|L(u)|=d+1$, and \emph{constrained} otherwise. 
The motivation to define $t$-feasible list assignments is to be able to treat the edge-colouring dynamics around a vertex, where some of the edges are incident to other precoloured ones. In this sense, constrained elements correspond to elements that may not be able to take any possible colour, and the parameter $t$ gives a lower bound on the number of choices. In our applications, it will be enough to consider $1$- and $2$-feasible list assignments, which will exist if we use $k\geq \Delta+1$ and $k\geq \Delta+2$ colours, respectively.

We should stress here that all the lists are subsets of $[d+1]$, although the results in this section can be generalised to arbitrary lists of size $k\geq d+1$ using a variant of Proposition~\ref{prop:mono} for list colouring. 
If $L$ is $t$-feasible, we have
\begin{align}\label{eq:size_L-col}
\left(\prod_{i=1}^t |L(u_i)|-i+1\right) (d+1-t)!\leq |\Omega^L_U| \leq \left(\prod_{i=1}^t |L(u_i)|\right)(d+1-t)!\;.
\end{align}

The chain $\cL^L_U$ is symmetric and it will follow from the results below that, if $L$ is $t$-feasible, for some $t$, then the chain is also ergodic. So its stationary distribution is uniform on $\Omega_U^L$.  The main goal of this section is to prove the following bound on its relaxation time.
\begin{lemma}\label{lem:rel_time_clique}
	If $L$ is $t$-feasible and $k=d+1$, then we have
	$$
	\tau(\cL^L_U)= O\left(\sum_{i=1}^t \frac{d}{p_i} + \sum_{i=t+1}^d \frac{1}{p_i}\right)\;.
	$$
\end{lemma}

In order to prove the lemma we will use the comparison techniques introduced in Section~\ref{sec:comparison}. Consider the dynamics $\cL_{\text{unif}}^L$ on $\Omega_U^L$ with generator matrix given for any $\sigma\neq \eta$ by
\begin{align}\label{eq:LU}
\cL_{\text{unif}}^L[\sigma\to \eta]&= \frac{1}{|\Omega_U^L|}\;.
\end{align}

Clearly $\tau(\cL_{\text{unif}}^L)=1$. The main idea will be to compare $\tau(\cL_U^L)$ with $\tau(\cL_{\text{unif}}^L)$. For technical reasons, it will be easier to introduce an intermediate chain and compare both to it. Define $\cL^L_{\text{int}}$ to be the dynamics on $\Omega^L_U$ with generator matrix given for any $\sigma\neq \eta$ by
\begin{align}\label{eq:LI}
\cL^L_{\text{int}}[\sigma\to \eta]&=
\begin{cases}
\frac{1}{|\Omega_U^L|} & \text{if }(\sigma,\eta) \text{ is a good pair},\\
0& \text{otherwise.}
\end{cases}
\end{align}
Speaking informally, the dynamics $\cL^L_{\text{int}}$ can be seen as $\cL_{\text{unif}}^L$ where only moves between ``good'' pairs of colourings are allowed. Informally speaking, a pair of colourings $(\alpha,\beta)$ is good if it contains no set $S$ of constrained vertices such that $\alpha(S)=\beta(S)$. We defer the formal definition of a good pair to later in the section. 
\medskip

Lemma~\ref{lem:rel_time_clique} follows from combining these two lemmas (proved in the next two subsections) with the fact that $\tau(\cL_{\text{unif}}^L)=1$.
\begin{lemma}\label{lem:LUvsLI}
If $k=d+1$, then for $t$ constant we have
	$$ 
	\frac {\tau(\cL_U^L)} {\tau(\cL^L_{\text{int}})} = O\left(\sum_{i=1}^t \frac{d}{p_i} + \sum_{i=t+1}^d \frac{1}{p_i}\right) \;.
	$$
\end{lemma}
\begin{lemma}\label{lem:LIvsLU}
If $L$ is $t$-feasible and $k=d+1$, then for $t$ constant we have
	$$ 
	\frac {\tau(\cL^L_{\text{int}})} {\tau(\cL_{\text{unif}}^L)} = O(1) \;.
	$$
\end{lemma}

\subsection{Comparing $\cL_U^L$ with $\cL^L_{\text{int}}$: the proof of Lemma~\ref{lem:LUvsLI}}\label{sec:clique_1}

Let $w$ be an additional vertex with $L(w)=[k]$, so $w$ is free. Consider an order on $U\cup \{w\}$ where $w$ is the smallest vertex. We can extend $\alpha\in \Omega_U^L$ to $U \cup \{w\}$, by letting $\alpha(w)$ be the unique colour not in $\alpha(U)$.
To every pair $\alpha,\beta\in \Omega_U^L$, we can assign a permutation $f=f(\alpha,\beta)$ on $U\cup \{w\}$ such that $f(x)=y$ if and only if $\alpha(x)=\beta(y)$. In particular, if $\alpha=\beta$ then $f$ is the identity permutation.
One can see the permutation $f$ as a blocking permutation: if $v=f(u)\in U$ for some $u\in U$, then $u$ blocks $v$ from being directly recoloured from $\alpha(v)$ to $\beta(v)$. However, if $v=f(w)\in U$, then $v$ can be directly changed to $\beta(v)$ in the colouring $\alpha$. 

It is useful to think about $f$ using its representation as a union of directed cycles. Throughout this section, by cycle in the permutation we mean a cycle of order at least $2$. If a cycle in $f$ contains~$w$, we will be able to recolour every vertex in the cycle by successively recolouring the vertices whose preimage is $w$. 
The main difficulty will arise from handling the other cycles that do not contain $w$, which correspond to circular blockings of vertices in $U$ (e.g. $u,v\in U$, $u$ blocks $v$ and $v$ blocks $u$). In this case, we will need to insert $w$ into this cycle, and then process it. The merging operation corresponds to recolouring a vertex in $\alpha$ with the only available colour, which must be in its list. This motivates the following classification: a cycle in $f$ is a \emph{type-1-cycle} if it contains $w$, a \emph{type-2-cycle} if it does not contain $w$ but has at least one free vertex and a \emph{type-3-cycle} if it does not contain $w$ and all its vertices are constrained. 

Recolouring a type-$3$-cycle $C$ might be hard because the only colours available in $L(C)$ might already be taken by other vertices in the clique. 
This motivates the definition of good pairs which governs~$\cL^L_{\text{int}}$.

\begin{defn}
A pair of colourings $(\alpha,\beta)$ is \emph{good} if $f(\alpha,\beta)$ contains no type-3-cycles.
\end{defn}
Let $(\alpha,\beta)$ be a good pair and $f=f(\alpha,\beta)$.
A \emph{$v$-swap} is an operation on $f$ that gives the permutation $\hat f$ obtained from $f$ by reassigning $\hat f(w)=f(v)$ and $\hat f(v)=f(w)$. These operations are in bijection with the valid transitions of $\cL^L_U$.

One can define recolouring sequences between $\alpha$ and $\beta$ using swaps. 
The order on $U\cup\{w\}$ gives a canonical way to deal with the cycles in $f(\alpha,\beta)$, processing the cycle with the smallest free vertex, at a time. 
Precisely, while there is a free vertex in a cycle of length at least $2$ of $f$, let $v$ be the smallest one and
\begin{itemize}
\item[-] if $v$ is in a type-$1$-cycle, then $v=w$. While $f(w)\neq w$, update $f$ by performing an $f(w)$-swap.
\item[-] if $v$ is in a type-$2$-cycle, then $v\in U$. Update $f$ by performing a $v$-swap. Then, while $f(w)\neq w$, update $f$ by performing an $f(w)$-swap.
\end{itemize}
As there are no type-$3$-cycles, after termination the procedure produces the identity permutation. As swaps correspond to valid recolouring moves, it gives a recolouring path $\gamma_{\alpha,\beta}$ from $\alpha$ to $\beta$ that uses transitions from $\cL_U^L$. Note that any vertex in $U$ is recoloured at most twice. In fact, a vertex is only recoloured twice if it is the smallest free vertex in a type-$2$-cycle.

For every transition $(\sigma,\eta)\in \cL_U^L$, define
$$
\Lambda_{\sigma,\eta}= \{(\alpha,\beta)\in \cL_{\text{int}}^L :\, (\sigma,\eta)\in \gamma_{\alpha,\beta}\}\;.
$$

Our goal is to prove Lemma~\ref{lem:LUvsLI} using the weighted canonical paths method from Section~\ref{sec:comparison}. To this end, the two following lemmas will help us analyse the congestion resulting from this construction.

\begin{lemma}\label{lem:at_most_2}
Given a transition $(\sigma,\eta)\in \cL_U^L$ and a permutation $f$, there are at most two good pairs $(\alpha,\beta)$ such that $f=f(\alpha,\beta)$ and $(\alpha,\beta)\in \Lambda_{\sigma,\eta}$. 
\end{lemma}

\begin{proof}
Let  $v$ be the vertex at which $\sigma$ and $\eta$ differ. The order in which we recolour the vertices is fixed by the permutation $f$. Let $v_1, \ldots, v_\ell$ be the vertices in the order they are recoloured with repetitions, where $v_i$ is the vertex recoloured at the $i$-th step. Given this sequence, at every step the only valid transition is to perform a $v_i$-swap. If the recolouring done by the transition  $(\sigma, \eta)$ corresponds to the $i^*$-th step in the sequence, then $\alpha$ and $\beta$ are fully determined. Indeed, $\beta$ can be recovered from $\eta$ by sequentially performing $v_i$-swaps for every $i>i^*$. Symmetrically, $\alpha$ can be recovered from $\sigma$ by performing $v_i$-swaps for every $i<i^*$ (recall that the $u$-swap operation is an involution).

Since every vertex is recoloured at most twice in a recolouring path, $v$ appears at most twice in the sequence $v_1, \ldots, v_\ell$. So there are at most $2$ choices for $i^*\in [\ell]$ with $v_{i^*}=v$, and by the argument above, there are at most two pairs $(\alpha, \beta)$, with $f(\alpha, \beta) = f$, and $(\alpha,\beta)\in \Lambda_{\sigma,\eta}$. \qedhere

\end{proof}

Using the previous lemma, we can bound the size of $\Lambda_{\sigma,\eta}$
\begin{lemma}
	\label{lem:countPaths}
	Let $t$ be a fixed constant. Suppose that $\sigma$ and $\eta$ differ at $v\in U$. If $v$ is free, then $|\Lambda_{\sigma,\eta}|= O\left(|\Omega_U^L|\right)$.
If $v$ is constrained, then	$|\Lambda_{\sigma,\eta}| = O(d |\Omega_U^L|)$.
\end{lemma}

\begin{proof}
By Lemma~\ref{lem:at_most_2}, it suffices to bound the number of permutations $f$ which are compatible with $(\sigma,\eta)$. The key idea is the following claim that not all candidates for $f$ are compatible, as constrained vertices can only block and be blocked by vertices with colours in their lists.
\begin{cla}
	Let $(\alpha,\beta)\in \Lambda_{\sigma,\eta}$ with $f(\alpha,\beta)=f$. For every constrained vertex $u\in U$ with $u\neq v$, one of the following holds:
	\begin{enumerate}[label=\roman*),start=1]
		\item\label{item:const_i} $\eta(f(u)) = \beta(f(u))$;
		\item\label{item:const_ii} $\sigma(f^{-1}(u)) = \alpha(f^{-1}(u))$;
	\end{enumerate}
	Moreover,~\ref{item:const_i} holds if and only if $u$ is recoloured in $\gamma_{\alpha,\beta}$ before the transition $(\sigma, \eta)$.
\end{cla}
\begin{proof}
	By construction of the recolouring paths, constrained vertices are only recoloured once. Let $C$ be the cycle containing $u$. If $v \not \in V(C)$, then the vertices in $C$ are recoloured either all before $v$, or all after $v$. In particular we have $\sigma|_{V(C)} = \eta|_{V(C)}$, and it is equal to $\beta|_{V(C)}$ if $u$ is recoloured before $(\sigma,\eta)$ and to $\alpha|_{V(C)}$ otherwise. Since $f(u), f^{-1}(u) \in V(C)$, either~\ref{item:const_i} or~\ref{item:const_ii} holds.
	
	If $v \in V(C)$, then 
	\begin{itemize}
		\item[-] If $f(u)$ is not the  smallest free vertex in $C$, then assume that~\ref{item:const_i} does not hold. Since $f(u)$ is recoloured only once, this means that $\eta(f(u)) = \alpha(f(u))$. Consequently, $f(u)$ is recoloured after the transition $(\sigma, \eta)$, and since $u \neq v$, this is also the case for $u$. This means that $f^{-1}(u)$ is recoloured either during or after the transition $(\sigma, \eta)$, and in both cases~\ref{item:const_ii} holds.
	
		\item[-] Symmetrically, if $f^{-1}(u)$ is not the smallest free vertex in $C$, then assuming that~\ref{item:const_ii} does not hold we conclude that $u$ has been coloured before the transition $(\sigma,\eta)$ and that~\ref{item:const_i} holds.

		\item[-] Finally, if $f(u) = f^{-1}(u)$ is the smallest free vertex in $C$, then $v = f(u)$. In this case, either this is the first time $v$ is recoloured, so $u$ is recoloured after $(\sigma,\eta)$ and $\sigma(v) = \alpha(v)$, and~\ref{item:const_ii} holds, or it is the second time $v$ is recoloured, so $u$ is recoloured before $(\sigma,\eta)$ and $\eta(v) = \beta(v)$ and~\ref{item:const_i} holds.
	\end{itemize}
\end{proof}

Suppose that $v$ is free. As there are at most $t$ constrained vertices, there are at most $2^t$ choices to decide which of items~\ref{item:const_i} or~\ref{item:const_ii} holds for them. For $s\in [t]$, let $x_1,\dots,x_s$ be the constrained vertices that satisfy~\ref{item:const_i} and $y_1,\dots,y_{t-s}$ be the ones that satisfy~\ref{item:const_ii}.
For every $i\in [s]$, by definition of $f$ we have $\eta(f(x_i))=  \beta(f(x_i))= \alpha(x_i)\in L(x_i)$ and there are at most $|L(x_i))|$ choices $z\in U\cup \{w\}$ for $f(x_i)=z$, namely the ones with $\eta(z)\in L(x_i)$. Analogously, for $i\in [t-s]$, by definition of $f$ we have $\sigma(f^{-1}(y_i))=\alpha(f^{-1}(y_i))=\beta(y_i) \in L(y_i)$ and there are at most $|L(y_i)|$ choices $z\in U\cup\{w\}$ for $f(z)=y_i$. 
Thus, there are at most $\prod_{i=1}^s |L(x_i)|\prod_{i=1}^{t-s} |L(y_i)|= \prod_{i=1}^t |L(u_i)|$ choices for the images of $x_i$ and the preimages of $y_i$. Note that $f(x_i)\neq y_j$ for every $i,j$ as~\ref{item:const_i} holds if and only if $u$ is recoloured before $(\sigma,\eta)$ and $v$ is a free vertex. So these choices fix exactly $t$ images in $f$.
Finally, there are at most $(d+1-t)!$ ways to complete $f$ by choosing successively the image of the remaining elements. Using Lemma~\ref{lem:at_most_2} and~\eqref{eq:size_L-col}, we have
$$
|\Lambda_{\sigma,\eta}|\leq 2\cdot 2^t (d+1-t)! \prod_{i=1}^t (|L(u_i)|) \leq 2^{t+1} \left(\prod_{i=1}^t \frac{|L(u_i)|}{|L(u_i)|-i+1}\right) |\Omega^L_U|  = O(|\Omega^L_U|)\;.
$$
Assume now that $v$ is constrained. There are at most $2^{t-1}$ ways to choose a configuration of~\ref{item:const_i} and~\ref{item:const_ii} for the remaining constrained vertices. In comparison to the case where $v$ is free, as neither \ref{item:const_i} nor~\ref{item:const_ii} holds for $u=v$, there is an extra factor $|L(v)|\leq d+1$ to choose either the image or the preimage of $v$ in $f$. Similarly as before, it follows that
$$
|\Lambda_{\sigma,\eta}|= O(d |\Omega_U^L|)\;.
$$
\end{proof}

We are now in a good situation to prove Lemma~\ref{lem:LUvsLI}.

\begin{proof}[Proof of Lemma~\ref{lem:LUvsLI}]
	
In order to apply the weighted canonical paths theorem to $\cL_U^L$ and $\cL^L_{\text{int}}$, we need to choose a weight function $\omega$ for all $(\sigma, \eta)\in \cL_U^L$. Let $v$ be the vertex where $\sigma$ and $\eta$ differ. We define $\omega$ as follows
\begin{align}\label{eq:def_weight}
		\omega(\sigma,\eta)&:= 
		\begin{cases}
		d/p_i &\text{ if $v=u_i$ for }i\in [t],\\
	   1/p_i &\text{ if $v=u_i$ for }i\in [d]\sm [t].
		\end{cases}
\end{align}		
As for every recolouring path $\gamma$, each element of $U$ is recoloured at most twice, it follows that
		\begin{align}\label{eq:weight_path}
	|\gamma|_\omega \leq 2\left(\sum_{i=1}^t \frac{d}{p_i} + \sum_{i=t+1}^d \frac{1}{p_i}\right)\;.
	\end{align}
Both stationary distributions of $\cL_U^L$ and $\cL_{\text{int}}^L$ are uniform on $\Omega_U^L$. Also recall that $\cL^L_{\text{int}}[\alpha\to \beta]=1/|\Omega_U^L|$ for any good pair $(\alpha,\beta)$.

Using Lemma~\ref{lem:countPaths}, regardless of whether the vertex where $\sigma$ and $\eta$ differ is free or constrained, we can bound the congestion of the transition $(\sigma,\eta)$ as follows
\begin{align*}
\rho_{\sigma,\eta} &= \frac{1}{\cL^L_U[\sigma\to \eta]\omega(\sigma,\eta)}\sum_{(\alpha, \beta)\in \Lambda_{\sigma,\eta}} \frac {|\gamma_{\alpha,\beta}|_\omega} {|\Omega_U^L|} \\
 &\leq \frac{|\Lambda_{\sigma, \eta}|}{\cL^L_U[\sigma\to \eta]\omega(\sigma,\eta)|\Omega_U^L|} \cdot \max_{\alpha,\beta} |\gamma_{\alpha,\beta}|_\omega  \\
 &= O(\max_{\alpha,\beta} |\gamma_{\alpha,\beta}|_\omega) =O\left(\sum_{i=1}^t \frac{d}{p_i} + \sum_{i=t+1}^d \frac{1}{p_i}\right)\;. 
\end{align*}
The desired result follows from Proposition~\ref{prop:weighted_cano_path}.
\end{proof}

\subsection{Comparing $\cL^L_{\text{int}}$ with $\cL^L_{\text{unif}}$: the proof of Lemma~\ref{lem:LIvsLU}}\label{sec:clique_2}

The proof of this lemma uses the fractional paths method with uniform weights. To define the paths, we first split the set of constrained vertices into two subsets. Let $A$ be the set of constrained vertices $u$ satisfying $|L(u)|\leq 2(3t+2)$, and let $B$ be the remaining ones. Before we define the paths, we will need the following result.
\begin{lemma}
	\label{lem:typeIII}
	Let $\alpha,\beta\in \Omega_U^L$ and let $\xi_A$ be an $L$-colouring of the vertices in $A$. Assume that $\alpha|_A$ and $\xi_A$ differ on at most one vertex, and similarly for $\xi_A$ and $\beta|_A$. There exists a constant $c(t)>0$ such that there are at least $c(t)|\Omega_U^L|$ colourings $\xi\in \Omega_U^L$ satisfying that $\xi|_A = \xi_A$, and both $(\alpha,\xi)$ and $(\xi, \beta)$ are good pairs. 
\end{lemma}
\begin{proof}
	Assume that $A=\{u_1, \ldots, u_a\}$, so $|B|=t-a$. Construct the colouring $\xi$ by setting $\xi|_A=\xi_A$ and then choosing the colours in $U\setminus A$ one by one, starting from the vertices in $B$, as follows:
	\begin{itemize}
		\item[-] for each $v\in B$, choose $\xi(v) \notin \xi(A)\cup \alpha(A\cup B)\cup \beta(A\cup B)$ that has not been used already; 
		\item[-] for each free vertex, choose a colour not used by the vertices already coloured in $\xi$.
	\end{itemize}
	When we choose the colour for $v \in B$, there are at most $3t+2$ forbidden colours. Indeed, there are at most $a+2(t-a)+2$ colours in $\xi(A)\cup \alpha(A\cup B)\cup \beta(A\cup B)$, and at most $t-a$ colours used by the previous vertices in $B$ that have already been coloured by $\xi$. Thus, for $v\in B$ there are at least $|L(v)| - (3t+2)\geq |L(v)|/2$ choices for $\xi(v)$. Using~\eqref{eq:size_L-col}, the total number of colourings extending $\xi_A$ is at least
	\begin{align*}
	\prod_{v\in B} (|L(v)| - (3t+2)) \cdot (d+1 - t)! &\geq \frac {(d+1 - t)! } {2^{t-a}} \prod_{v\in B} |L(v)|   \\
	&\geq \frac {(d+1 - t)! } {2^{t} (3t+2)^a} \prod_{i=1}^t |L(u_i)| \\ 
	&\geq  \frac {|\Omega_U^L|} {(2(3t+2))^t}\;.
	\end{align*}
	It suffices to show that for any such extension $\xi$ of $\xi_A$, $(\alpha, \xi)$ and $(\xi, \beta)$ are good pairs. We only prove it for $(\alpha, \xi)$ as the other case is symmetric. Assume by contradiction that $f(\alpha,\xi)$ contains a type-$3$-cycle, and let $C$ be this cycle. Then $V(C) \cap B = \emptyset$. Indeed, if $v \in V(C)\cap B$, then there exists a vertex $u\in A\cup B$ with $\alpha(u)=\xi(v)$, but this contradicts the fact that $\xi(v)\notin \alpha(A\cup B)$. Thus, $V(C) \subseteq A$, but this is not possible since  $\alpha|_A$ and $\xi|_A$ differ by at most one vertex.
\end{proof}

We can now compare the relaxation times of $\cL^L_{\text{int}}$ and $\cL^L_{\text{unif}}$.
\begin{proof}[Proof of Lemma~\ref{lem:LIvsLU}] 

We use the fractional paths method. Note that both $\cL^L_{\text{int}}$ and $\cL^L_{\text{unif}}$ are ergodic, reversible and symmetric and that their stationary distributions are uniform on $\Omega_U^L$. 

It suffices to  define a collection of fractional paths $\Gamma_{\alpha,\beta}$ between any two colourings $\alpha$ and $\beta$ in~$\Omega_U^L$. Since there are at most $t$ constrained vertices and their lists have size at least $t+1$, we can find a sequence of $L$-colourings of $A$, $\alpha|_A=\xi_A^0, \xi_A^1, \dots ,\xi_A^m = \beta|_A$, such that any two consecutive colourings differ by one vertex. Let $M$ be an upper bound on the length of these paths for every $\alpha,\beta$, which only depends on $t$. 

For $m\in [M]$, let $\Gamma^{m}_{\alpha, \beta}$ be the collection of all the paths of the form $\alpha=\xi^0,\xi^1,\dots,\xi^m = \beta$ where $\xi^i|_A=\xi^i_A$, and $(\xi^{i-1},\xi^i)$ is a good pair for $i\in [m]$. By Lemma~\ref{lem:typeIII}, for each $i\in [m-1]$ there are at least $c(t) |\Omega_U^L|$ choices for $\xi^i$, independently of the choices of the other $\xi^j$ for $j \neq i$. Thus, $|\Gamma^{m}_{\alpha, \beta}| \geq (c(t)|\Omega_U^L|)^{m-1} $, and if $g_{\alpha,\beta}$ is the uniform flow from $\alpha$ to $\beta$, then each $\gamma\in \Gamma^m_{\alpha, \beta}$ satisfies $g_{\alpha,\beta}(\gamma)\leq (c(t) |\Omega_U^L|)^{-(m-1)}$. Let $\Gamma^m= \{\Gamma^{m}_{\alpha,\beta}:\, \alpha,\beta\in \Omega_U^L\}$ and $\Gamma=  \cup_{m=1}^M\Gamma^{m}$.
	
We need to bound the congestion of any good pair $(\sigma,\eta)$. We fix $m\in [M]$ and will bound the contribution of $\Gamma^m$ to it. Let $\gamma= \xi^0,\xi^1,\dots,\xi^m \in \Gamma^m$ containing $(\sigma,\eta)$. 
There are at most $m$ choices for $i\in [m]$ such that $\sigma=\xi^{i-1}$ and $\eta=\xi^i$. Then, there are at most $|\Omega_U^L|^{m -1}$ choices for $\xi^j$ with $j\notin \{i-1,i\}$. Each such path satisfies $g(\gamma)\leq  (c(t)|\Omega_U^L|)^{-(m-1)}$, for any uniform flow $g$. Thus,
$$
\sum_{\gamma\in \Gamma^{m}\atop \gamma\ni (\sigma,\eta)}g(\gamma)|\gamma| \leq m  c(t)^{-(m-1)}\;.
$$
Hence,
$$
\rho_{\sigma,\eta} \leq \sum_{m=1}^M m c(t)^{-(m-1)}= O(1)\;,
$$
and, by Proposition~\ref{prop:frac_path}, we obtain the desired result.
\end{proof}

\subsection{Dynamics of two cliques intersecting at a vertex}
\label{sec:biclique}

In this section we study a similar dynamics, that we will also use in the main proof. Let $z$ be a vertex. For $d \geq 1$, let $X=\{z,x_1,\dots, x_{d-1}\}$ and $Y=\{z,y_1,\dots,y_{d-1}\}$ two sets of vertices and consider the graph with vertex set $Z=X\cup Y$ where each set $X$ and $Y$ induces a clique. As before, we fix the number of colours $k=d+1$. 
For a list assignment $L$ of $Z$, recall the definition of $\Omega_Z^L$ and the dynamics $\cL^L_Z$, where we denote by $p_z,p_1,\dots,p_{d-1},q_1,\dots,q_{d-1}$ the parameters for $z, x_1,\dots,x_{d-1},y_1,\dots,y_{d-1}$, respectively.

Let $t, t_X,t_Y$ be non-negative integers with $t\geq t_X+t_Y$ and $1 \leq t_X,t_Y\leq d-1$. Without loss of generality, we will assume that $d$ is sufficiently large with respect to $t$. If this is not the case, then $|\Omega_Z^L|$ is a constant depending on $t$, and the relaxation time is $O(1)$.

A list assignment $L$  is \emph{$(t,t_X,t_Y)$-feasible} if 
\begin{itemize}
\item[-] $|L(z)|=d+1$;
\item[-] $|L(x_i)|,|L(y_j)|\geq t$ for every $i\in [t_X]$ and $j\in [t_Y]$;
\item[-] $|L(x_i)|,|L(y_j)|=d+1$ for every $i\in [d-1]\sm [t_X]$ and $j\in [d-1]\sm [t_Y]$;
\end{itemize}We define free and constrained vertices as before, with the exception of $z$ which is considered a constrained vertex.
If $L$ is $(t,t_X,t_Y)$-feasible, then
\begin{align}\label{eq:size_L-col2}
 |\Omega^L_Z| &\geq \Big(\prod_{i=1}^{t_X} |L(x_i)|-i+1\Big)\Big(\prod_{j=1}^{t_Y} |L(y_j)|-j+1\Big) \cdot (d+1-t) (d-t_X)! (d-t_Y)!\\
 |\Omega^L_Z| &\leq \Big(\prod_{i=1}^{t_X} |L(x_i)|\Big)\Big(\prod_{j=1}^{t_Y} |L(y_j)|\Big)\cdot  d(d-t_X)! (d-t_Y)!\;.
\end{align}
As before, the chain is symmetric and, if $L$ is $(t,t_X,t_Y)$-feasible for some $t,t_X,t_Y$ satisfying the conditions stated above, it follows from results below that it is ergodic, so its stationary distribution is uniform.

We will prove a bound analogous to the one in Lemma~\ref{lem:rel_time_clique} on the relaxation time of $\cL^L_{Z}$.
\begin{lemma}\label{lem:rel_time_clique2}
	If $L$ is $(t,t_X,t_Y)$-feasible and $k=d+1$, then for $t$ constant we have
	$$
	\tau(\cL^L_{Z})= O\left(\frac{d^2}{p_z}+ \sum_{i=1}^{t_X} \frac{d}{p_i} + \sum_{i=t_X+1}^{d-1} \frac{1}{p_i}+ \sum_{j=1}^{t_Y} \frac{d}{q_j} + \sum_{j=t_Y+1}^{d-1} \frac{1}{q_j}\right)\;.
	$$
\end{lemma}

The proof follows the same lines as the proof of Lemma~\ref{lem:rel_time_clique}, so we will only sketch it, stressing the parts where the two differ. For $\alpha\in \Omega_{Z}^L$, denote by $\alpha_X$ and $\alpha_Y$ the restrictions of $\alpha$ onto $X$ and $Y$, respectively. As we did for the clique, we extend $\alpha$ by adding two artificial vertices: $w_X$ in $X$ and $w_Y$ in $Y$, and by assigning to them the only available colour in each set. Define the permutations $f_X$ and $f_Y$ as before.
 
We say that $(\alpha,\beta)$ is good if and only if both $(\alpha_X,\beta_X)$ and $(\alpha_Y,\beta_Y)$ are good (i.e., there is no type-$3$-cycle in the permutations $f_X$ and $f_Y$). Remember that $z$ is a constrained vertex, and as a consequence if $(\alpha, \beta)$ are good, every cycle in the permutations contains a free vertex different from $z$. 
Redefine the dynamics $\cL^L_{\text{unif}}$ on $\Omega_Z^L$ as in~\eqref{eq:LU}, and, using the new definition of good pairs, redefine $\cL^L_{\text{int}}$ on $\Omega_Z^L$ as in~\eqref{eq:LI}. We proceed in two steps by bounding the ratios of the relaxation times of $\cL^L_Z $ and $\cL^L_{\text{int}}$ and of $ \cL^L_{\text{int}}$ and $\cL^L_{\text{unif}} $ .
\begin{lemma}
	\label{lem:LUvsLI2}
	If $k=d+1$, then for $t$ constant we have
	$$\frac {\tau(\cL^L_Z)} {\tau(\cL^L_{\text{unif}})} = O\left(\frac{d^2}{p_z}+ \sum_{i=0}^{t_X} \frac{d}{p_i} + \sum_{i=t_X+1}^{d-1} \frac{1}{p_i}+ \sum_{j=0}^{t_Y} \frac{d}{q_j} + \sum_{j=t_Y+1}^{d-1} \frac{1}{q_j}\right)\;. $$
\end{lemma}

\begin{proof}[Sketch of the proof]
We reuse the recolouring paths defined in Lemma~\ref{lem:LUvsLI} for the clique. Given two colourings $\alpha$ and $\beta$, denote by $\gamma^X_{\alpha_X, \beta_X}$ and $\gamma^Y_{\alpha_Y, \beta_Y}$ the recolouring paths constructed for each of the two sets $X$ and $Y$ independently. Observe that for each path in these sets, each constrained vertex is recoloured at most once. In particular, $z$ changes its colour at most once. Construct the recolouring path $\gamma_{\alpha, \beta}$ in the following way:
\begin{itemize}
	\item[-] apply the recolourings in $\gamma^X_{\alpha_X,\beta_X}$ until $z$ needs to be recoloured;
	\item[-] apply the recolourings in  $\gamma^Y_{\alpha_Y,\beta_Y}$;
	\item[-] apply the remaining recolourings in  $\gamma^X_{\alpha_X,\beta_X}$.
\end{itemize}
Note that in the second step, $z$ can be safely recoloured with $\beta(z)$ because its target colour is available in $X$, since the next move according to $\gamma^X_{\alpha_X,\beta_X}$ would be to recolour $z$ with colour $\beta(z)$. We need to bound the congestion of each transition for this collection of paths. 

For a transition $(\sigma,\eta)$, let $\Lambda_{\sigma,\eta}$ be the set of good pairs $\alpha,\beta$ such that $\gamma_{\alpha,\beta}$ contains $(\sigma,\eta)$.
The analogues of Lemmas~\ref{lem:at_most_2} and~\ref{lem:countPaths} hold in this setting. 
In particular, for every $(\sigma,\eta)$ differing at a vertex $v$, if $v$ is free then $|\Lambda_{\sigma,\eta}|= O\left(|\Omega_Z^L|\right)$, if $v\neq z$ is constrained then $|\Lambda_{\sigma,\eta}| = O(d |\Omega_Z^L|)$, and if $v=z$ then $|\Lambda_{\sigma,\eta}| = O(d^2 |\Omega_Z^L|)$ as we get an extra factor $d$ for each permutation. This allows us to bound the congestion of a transition as in the previous section, and so Lemma~\ref{lem:LUvsLI2} follows.
\end{proof}

\begin{lemma}\label{lem:LIvsLUbis}
	If $L$ is $(t,t_X,t_Y)$-feasible and $k=d+1$, then for $t$ constant we have
	$$ 
	\frac {\tau(\cL^L_{\text{int}})} {\tau(\cL_{\text{unif}}^L)} = O(1) \;.
	$$
\end{lemma}
\begin{proof}[Sketch of the proof]
The lemma can be proved using the same steps as in the proof of Lemma~\ref{lem:LIvsLU}. Let $A_X$ and $A_Y$ be the set of constrained vertices with lists of size at most $2(3t+2)$ in $X$ and $Y$, respectively. Let $A=A_X\cup A_Y$ and let $B$ be the set of constrained vertices that are not in $A$. As $z$ is a constrained vertex with $|L(z)|=d+1 > 2(3t+2)$, we have $z\in B$. The analogous of Lemma~\ref{lem:typeIII} still holds in this setting.

For any $\alpha,\beta\in \Omega_Z^L$, we would like to find a sequence $\alpha|_A = \xi^0_A, \dots, \xi^m_A= \beta|_A$ of colourings of $A$ such that each consecutive pair differs only at one vertex. As $z\not\in A$ and as all constrained vertices in $A$ have a list of size at least $t+1$, this sequence can be found by  independently recolouring $A_X$ and $A_Y$.
Arguing as in the proof of Lemma~\ref{lem:LIvsLU}, we can construct many recolouring paths between $\alpha$ and $\beta$ with transitions that correspond to good pairs. Then, Proposition~\ref{prop:frac_path} implies the desired result.
\end{proof}

\section{Glauber dynamics on edge-colourings of trees}\label{sec:gen}

In this section we prove our main theorem. We follow a similar approach to the one of Lucier and Molloy in~\cite{LucMol} for vertex-colourings, by recursively splitting the tree into smaller subtrees using block dynamics. However, there are several points where our strategies differ.

\subsection{Relaxation time of block dynamics}
\label{sec:relaxBlockTree}

In all this section, we will assume that $G=({V_G},{E_G})$ is a $d$-regular tree, that is every internal vertex has degree \emph{exactly} $d$. We also assume that $k=d+1$. 
\begin{defn}
A subtree $T$ of $G$ is splitting if one of the following holds.
\begin{itemize}
\item[-] $T$ is a single edge, 
\item[-] $T$ has fringe interior boundary $|\overline\partial T| \leq 2$ and if $|\overline\partial T|=2$ with $\overline\partial T=\{e,f\}$, then $e$ and $f$ are not incident.
\end{itemize} 
\end{defn}
Fix $\mu\in \Omega_{E_G}$, and fix $T = (V, E)$ a splitting subtree of $G$. If $T$ is not a single edge, then it has fringe boundary, so it is $d$-regular. The central point of our proof is to study $\cL_{E}^\mu$ by decomposing it into the dynamics of its subtrees using the block dynamics defined in Section~\ref{sec:block}. We will assume that $T$ is rooted in one of the two following ways.

\paragraph{Vertex-rooted trees:} The root of $T$ is  $r\in V$, an internal vertex of $T$. Let $e_1,\dots,e_d$ be the edges incident to $r$. For each $i\in [d]$, we consider the block formed by the edges of the subtree hanging from $e_i$ (inclusive of $e_i$).

\paragraph{Edge-rooted trees:} The root of $T$ is an edge $e=xy$ where $x$ and $y$ are internal vertices of $T$ (so, $e\notin \overline\partial T$).
Let $e_1,\dots,e_{2d-2}$ be the edges incident to $e$. We let $\{e\}$ be a block, and for every $i\in [2d-2]$, we consider the block formed by the edges of the subtree hanging from $e_i$.

\medskip

In both cases, we denote by $\cE=\{E_1, \ldots, E_r \}$ the block partition  described above. Note that each block in $\cE$ contains at most one edge incident to edges in other blocks. Let $H=(U,F)$ be the subgraph induced by these edges. Note that in the case of a vertex rooted tree, $H$ is a star, and in the case of an edge rooted tree, $H$ is a bi-star: two stars joined by an edge. For each $i \in [r]$, let $T_i$ be the subtree with edge set $E_i$. 

Throughout this section we will make the following two assumptions on $\cE$.
\begin{enumerate}[label=(A\arabic*)]	
	\item\label{item:A1} $T_i$ is splitting for every $i\in [r]$;
	\item\label{item:A2}  $\cL_{E_i}^\sigma$ is ergodic for every $i\in [r]$ and every $\sigma \in \Omega_{E}^\mu$.
\end{enumerate}

Thus, we can define the reduced block dynamics on $\cE$ with boundary condition $\mu$. Recall from Section~\ref{sec:block} that its state space is $\Omega_\cR$, the restriction to $H$ of the colourings in $\Omega_{E}^\mu$. Moreover, its transition matrix is given for any $\hat\sigma\neq  \hat\eta$: 
\begin{align}
\cR_{\cE}^\mu[\hat\sigma \rightarrow \hat\eta]&=\begin{cases}
g_i \pi^{i, \hat\sigma}_{\text{proj}}(\hat\eta) 
& \text{if there exists }i\in [r]\text{ such that }\hat\eta=\eta|_U\text{ for some }\eta\in \Omega_{E_i}^{\hat\sigma},\\
0& \text{otherwise.}
\end{cases}
\end{align}
where $\pi^{i, \hat\sigma}_{\text{proj}}$ is the projection of $\pi_{E_i}^{\hat\sigma}$ onto $F$~(see~\eqref{eq:proj_2}). To bound the relaxation time on $T$, we will proceed in two steps. First we compare the original dynamics $\cL_\cE^\mu$ to the reduced block dynamics on $\cE$ using Propositions~\ref{prop:mar} and~\ref{prop:red} from Section~\ref{sec:block}. Then, we bound the relaxation time for the reduced block dynamics using the results from  Section~\ref{sec:clique}. The three following lemmas will help us for the second step. They show that the transitions rates and the stationary distribution of the reduced block dynamics are close to uniform.
 We will first prove bounds on the stationary distribution for the reduced block dynamics, and then proceed to bound the relaxation time for the reduced block dynamics.

\begin{lemma}\label{lem:red_dyn}
	Assuming~\ref{item:A1}--\ref{item:A2}, the reduced block dynamics $\cR_{\cE}^\mu$ is ergodic and reversible, and its stationary distribution $\pi_\cR$ is the projection of $\pi_{E}^\mu$ onto $F$.
\end{lemma}
\begin{proof}
Let $H^\ell$ be the line graph with vertex set $F$, the set of edges of $H$. Consider the Glauber dynamics $\cL_ {F}^L$ with the following list constraints on $e\in F$
	\begin{itemize}
		\item $L(e) = [k] \setminus \mu(N(e) \cap \partial T)$ if $e \in \overline \partial T$,
		\item $L(e) = [k]$ for every other edge.
	\end{itemize}
Recall that $N(e)$ is the set of edges incident to $e$.	Then, since each block $E_i$ contains only one edge in $H$, the reduced block dynamics $\cR_\cE^\mu$ has exactly the same transitions as $\cL_ {F}^L$, but with possibly different probability transitions. Thus, $\cR_\cE^\mu$ is ergodic if and only if $\cL_ {F}^L$ is.
	
 If $T$ is vertex-rooted, then $H$ is a star, and  $H^\ell$ is a clique. Additionally, since $T$ is splitting, the two edges in $\overline \partial T$ are not adjacent, and in particular only one is in $H$. This edge, if it exists, is assigned a list of length $2$,  so $L$ is $1$-feasible. The ergodicity of $\cL_{F}^L$ follows from Lemma~\ref{lem:rel_time_clique} with $t\leq 1$.
		 
 If $T$ is edge-rooted, then $H$ is a bi-star, and $H^\ell$ is composed of two cliques intersecting at one vertex. Moreover, the two edges in $\overline \partial T$ cannot be in the same side of the bi-clique, and each has a list of length at least $2$. So $L$ is $(2,1,1)$-feasible. The ergodicity of $\cL_{F}^L$ follows from Lemma~\ref{lem:rel_time_clique2} with $t_X,t_Y\leq 1$ and $t\leq 2$.

Lemma~\ref{lem:reversible_red} implies that the reduced block dynamics is reversible for the projection of $\pi_{E}^\mu$ onto $F$, concluding the proof.
\end{proof}

Before giving a bound on the relaxation time of the reduced block dynamics, we will prove some bounds on its stationary distribution $\pi_\cR$ to show that it deviates from a uniform distribution by at most a constant factor. To this end, the following lemma is a technical tool that we will reuse later. It shows that, given a boundary configuration, under the uniform distribution the probability that an edge $e_i$ is assigned an available colour is close to uniform.
\begin{lemma}\label{lem:each_col}
	For any subtrees $T_i$ with edge-set $E_i$, any $e_i\in \overline{\partial} T_i$ and any $\sigma \in \Omega^\mu_E$, let $C=\sigma(N(e_i)\cap \partial T_i)$. Assuming~\ref{item:A1}--\ref{item:A2}, for every $c\in [k]\setminus C$, we have 
	$$
	\pi_{E_i}^{{\sigma}}(\{\xi\in \Omega_{E_i}^{\sigma}:\,  \xi(e_i)=c\}) = \frac{1}{2}(1+O(1/d))\;.
	$$ 
	Moreover, if $|\overline{\partial} T_i| = 1$,	then
	$$
	\pi_{E_i}^{\sigma}(\{\xi\in \Omega_{E_i}^{\sigma}:\,  \xi(e_i)=c\}) = 1/2\;.
	$$ 
\end{lemma}

\begin{proof}
	
	First assume that $\overline{\partial} T_i= \{e_i,f_i\}$. As $G$ is $d$-regular, and $T_i$ is splitting, in particular its boundary is fringe. Since $k =d+1$, this implies that there are exactly two colours available for $e_i$ and two colours for $f_i$.

	We will bound $|\Omega_{E_i}^\sigma|$ and $|\Omega_{E_i}^\sigma(c)|$, the number of colourings in $\Omega_{E_i}^\sigma$ that assign $c$ to $e_i$. Let $P=(V_P,E_P)$ be the unique path in $T_i$ that connects $e_i$ and $f_i$ , including both edges, and let $s=|V_P|$. As $T$ is splitting and $|\overline\partial T_i|=2$, we have $s\geq 4$. If we fix a colouring $\xi_P$ of $E_P$, observe that the number of colourings of $\xi \in \Omega_{E_i}^\sigma$, such that $\xi|_P = \xi_P$ is independent of $\xi_P$. Indeed, if we remove $E_P$, we obtain a collection of rooted subtrees $T'_1,\ldots, T'_{s-2}$, where $T'_j$ has root $v_j\in V_P$. Given $\xi_P$, there are exactly $(d-1)!$ ways to colour the edges of $T'_j$ incident to $v_j$, and for each internal vertex, there are exactly $d!$ ways to choose a colouring of the edges hanging from it.

Therefore, in order to bound the ratio $|\Omega_{E_i}^\sigma(c)| / |\Omega_{E_i}^\sigma|$, we only need to compute $|\Omega_{E_P}^\sigma|$, and $|\Omega_{E_P}^\sigma(c)|$, respectively the number of edge-colourings of $P$ compatible with $\sigma$, and the number of these colourings $\xi_P$ for which $\xi_P(e_i) = c$. 
We can obtain a colouring of $P$ by first colouring $e_i$ and $f_i$, and then choosing the colour of the other edges in $P$ in the order they appear on the path from $e_i$ to $f_i$. As $s\geq 4$, there is at least one edge in $E_P\sm\{e_i,f_i\}$. For each of these edges except for the last one, there are $d$ choices of colours. For the last edge there are either $d$ or $d-1$ choices.
	It follows that 
	\begin{align*}
	4d^{s-1}(d-1) &\leq |\Omega^{\sigma}_{E_P}| \leq 4 d^{s}\;,\\
	2d^{s-1}(d-1) &\leq 	|\Omega^{\sigma}_{E_P}(c)| \leq  2 d^{s} \;.
	\end{align*}
	We conclude that 
	$$
	\pi_{E_i}^{\sigma}( \{\xi\in \Omega_{E_i}^\sigma:\, \xi(e)=c\} ) = \frac {|\Omega_{E_i}^\sigma(c)|} {|\Omega_{E_i}^\sigma|} = \frac {|\Omega_{E_P}^\sigma(c)|} {|\Omega_{E_P}^\sigma|} =\frac{1}{2}(1+O(1/d))\;.
	$$
	The second statement follows by a simple symmetry argument.
\end{proof}

 \begin{lemma}\label{lem:stat_dist}
 	Assuming~\ref{item:A1}--\ref{item:A2}, for every $\hat\sigma \in \Omega_{\cR}$, we have
	$$
	\pi_{\cR}(\hat\sigma) = \frac{1+O(1/d)}{|\Omega_{\cR}| }\;. 
	$$
\end{lemma}
\begin{proof}
Recall that $H=(U,F)$ and that $e_i \in F\cap E_i$ is unique edge of $E_i$ in $H$. Let $\hat\sigma\in \Omega_{\cR}$, we will compute $\pi_{\cR}(\hat\sigma)$ by using Lemma~\ref{lem:red_dyn} and by bounding the number of $\sigma\in \Omega_{E}^\mu$ such that $\sigma|_{F} = \hat \sigma$. If $|\overline\partial T_i|=1$, then $e_i$ is the only boundary edge and, by symmetry, the number of extensions of $\hat\sigma$ in $E_i$ does not depend on $\hat \sigma(e_i)$. If $|\overline\partial T_i|=2$, then there exists $f\in E_i$ with $f\neq e_i$ such that $f\in \overline\partial T_i$. As in the proof of Lemma~\ref{lem:each_col}, in this case the number of extensions is the same, up to a $1+O(1/d)$ multiplicative factor. Since $T_i$ is splitting, there are at most two values of $i\in [r]$ with $|\overline\partial T_i|=2$. It follows that, up to a $1+O(1/d)$ multiplicative factor, each $\hat \sigma$ has the same number of extensions. This concludes the proof.
\end{proof}

We will also need the following simple bound on the gap of the dynamics of a single edge.
\begin{lemma}\label{lem:gap_edge}
	Let $e\in {E_G}$ be an edge of $G$. Then,
	$$
	\min_{\mu \in \Omega_{E_G}} \Gap (\cL_{\{e\}}^\mu) \geq \frac{2}{d+1}\;.
	$$
\end{lemma}
\begin{proof}
	Observe that $\Omega_{\{e\}}^\mu$ is the set of colourings of a single edge with $k_0\geq 1$ colours, where each transition happens at rate $1/(d+1)$. If $k_0=1$, then the relaxation time is $1$. If $k_0 \geq 2$, all positive eigenvalues of $-\cL_{\{e\}}^\mu$ are equal to $k_0/(d+1)$, so $\Gap(\cL_{\{e\}}^\mu)\geq 2/(d+1)$.
\end{proof}

We can finally obtain a bound on the relaxation time of $\cL_{E}^\mu$.
\begin{lemma}\label{lem:blockDyn} Assuming~\ref{item:A1}--\ref{item:A2}, the following holds,
$$
\tau(\cL_{E}^\mu)= O\left(d^3 + \sum_{i=1}^r \tau_i \right)\;,
$$
where $\tau_i  := \max_{\sigma \in \Omega_{E}^\mu} \tau(\cL_{E_i}^\sigma)$.
\end{lemma}
\begin{proof}

Using Propositions~\ref{prop:mar} and~\ref{prop:red}, we know that the relaxation time of $\cL_{E}^\mu$ satisfies
\begin{align}\label{eq:first}
 \tau(\cL_{E}^\mu) \leq \tau(\cR_{\cE}^\mu) \;.
 \end{align}
Thus, to get the result, we only need to bound the relaxation time of the reduced block dynamics with partition $\cE$. We define an alternative dynamics.
Let $\cR_{\text{const}}$ be the continuous-time Markov chain with state space $\Omega_{\cR}$ and generator matrix given for any $\hat\sigma \neq \hat\eta$ by
\begin{align*}
\cR_{\text{const}}[\hat\sigma \rightarrow \hat\eta]&=\begin{cases}
g_i & \text{if $\hat\sigma$ and $\hat\eta$ differ only at $e_i$},\\
0& \text{otherwise,}
\end{cases}
\end{align*}
where $g_i = 1/ {\tau_i}$. Observe that $\cR_\cE^\mu$ and $\cR_{\text{const}}$ have the same state space and transitions (but different transition probabilities). For every $\hat\sigma,\hat\eta\in  \Omega_\cR$, Lemma~\ref{lem:each_col} implies that $\pi_{\text{proj}}^{\hat\sigma,i} (\hat \eta)= \pi^{\sigma}_{E_i}(\{\xi \in \Omega^{\sigma}_{E_i}:\xi(e_i)=\hat\eta(e_i)\})=\Theta(1)$, where $\sigma$ is an arbitrary colouring in $\Omega_*^{\hat\sigma}$. So 
$\cR_\cE^\mu[\hat\sigma \rightarrow \hat\eta]=\Theta(\cR_{\text{const}}[\hat\sigma \rightarrow \hat\eta])$.

Moreover, the stationary distribution $\pi_{\text{const}}$ of $\cR_{\text{const}}$ is uniform on $\Omega_\cR$, and by Lemma~\ref{lem:stat_dist} we have $\pi_{\cR}(\hat\sigma)= \Theta(\pi_{\text{const}}(\hat\sigma) )$ for every $\hat\sigma\in \Omega_\cR$.
Thus, it follows from Proposition~\ref{prop:constant_dyn} that 
\begin{align}\label{eq:const}
\tau(\cR^\mu_{\cE})= \Theta(\tau(\cR_{\text{const}}))
\end{align}

We will bound the relaxation time of $\cR_{\text{const}}$ using the results in Section~\ref{sec:clique}, and conclude using~\eqref{eq:first} and~\eqref{eq:const}.

Suppose first that $T$ is vertex-rooted. This implies that $H$ induces a star with edges $e_1,\dots, e_d$. Consider the clique with vertex-set $U=\{u_1,\dots,u_d\}$, where $u_i$ is identified with the edge $e_i$, and the list assignment $L$ of $U$ defined by $L(u_i)=[k]\sm \mu(N(e_i)\cap \partial T)$. Up to relabelling of the edges, as $k= d+1$, the list assignment $L$ is $1$-feasible. Consider the dynamics $\cL_U^L$ with probabilities $p_i:=g_i$. We can identify $\Omega_\cR$ and $\cR_{\text{const}}$ with $\Omega^L_U$ and $\cL_U^L$.
By applying Lemma~\ref{lem:rel_time_clique} with $t \leq 1$ we have,
$$
\tau(\cR_{\text{const}}) = \tau(\cL^L_U) = O\left(\sum_{i=1}^t \frac{d}{g_i}+\sum_{i=t+1}^d\frac{1}{g_i}\right) = O\left(d^2+\sum_{i=1}^d \frac{1}{g_i}\right)\;.
$$
In the last equality we have used Lemma~\ref{lem:gap_edge} as for $i\in [t]$, $T_i$ splitting and $u_i$ constrained imply that $T_i$ is a single edge.

If $T$ is edge rooted, then $H$ is a bi-star. We can do the same proof replacing Lemma~\ref{lem:rel_time_clique} by Lemma~\ref{lem:rel_time_clique2} obtaining a similar bound with an additive factor of order $d^3$.

\end{proof}

\subsection{Proof of Theorem~\ref{thm:stronger}}
\label{sec:proofThm}

Let $k\geq \Delta+1$ and $G= ({V_G}, {E_G})$ be a tree on $n$ vertices with maximum degree at most $\Delta$. 
From~\eqref{eq:mixrel2} and~\eqref{eq:cont_disc} to prove Theorem~\ref{thm:stronger} it suffices to bound $\tau(\cL_{E_G})$.

Let $d:=k-1\geq \Delta$. Construct the \emph{$d$-regularisation} $G^d = (V^d_G, E^d_G)$ of $G$ by adding $d-|N(u)|$ leaves adjacent to each internal vertex $u\in V_G$. Note that $G^d$ is $d$-regular as a tree and it has at most $dn$ vertices. Moreover, as we only added leaves to the original tree, the neighbourhood of any added edge in the line graph is a clique of size $d-1=k-2$.
So we can apply Proposition~\ref{prop:mono} obtaining $\tau(\cL_{E_G})\leq \tau(\cL_{E^d_G})$.
To prove Theorem~\ref{thm:stronger} it suffices to bound the relaxation time of $d$-regular trees with at most $dn$ vertices, and in the following we assume that $G$ is $d$-regular.

\medskip
We will prove the following result by induction on the size of the subtree $T$: 
\begin{cla}\label{cla:1}
	There is a constant $C$ such that, for every splitting subtree $T = (V, E)$ of $G$ with $m$ edges, and every edge-colouring $\mu \in \Omega_{E_G}$, the Glauber dynamics $\cL_{E}^\mu$ with parameters $p_i=1/k$ is ergodic and
		$$ \tau(\cL_{E}^\mu) \leq d^3 m^C\;.   $$
\end{cla}
From this claim, the theorem is obtained immediately by taking $T = G$. If $T$ is composed of a single edge, the Claim~\ref{cla:1} follows from Lemma~\ref{lem:gap_edge}. Let $m$ be the number of edges in $T$ and assume $m > 1$. Let $v\in V$ be the vertex such that each subtree $T_i$ hanging from $T$ rooted at $v$ has at most $\lceil \frac m 2 \rceil$ edges; this vertex always exists and it is internal. Let $e_i$ be the edge from $T_i$ incident to $v$. We are going to split $T$ into several subtrees by applying Lemma~\ref{lem:blockDyn}, possibly several times. Note that in order to apply this lemma, we must ensure that each of the subtrees is splitting. This gives constraints on how we can split $T$. Precisely, while splitting the tree into subtrees, none of the subtrees can have an internal boundary of size at least $3$, and for the subtrees with an internal boundary of size $2$, the two edges in the boundary must be non-incident.

The splitting procedure is done according to different cases:
\begin{enumerate}
	
	\item \label{split:case1} If all the $T_i$ are splitting, then simply root $T$ at $v$, and apply Lemma~\ref{lem:blockDyn}. In the following we will assume that not all the $T_i$ are splitting.
	
	\item \label{split:case2} If $|\overline \partial T| = 1$, then there is exactly one subtree, say $T_1$, that is not splitting. In particular, $|\overline \partial T_1| = 2$ and the edge in $\overline \partial T$ is also an edge in $T_1$ and is incident to $e_1$, violating the non-incident condition. All other subtrees $T_i$ for $i \neq 1$ are splitting and $|\overline \partial T_i| =1$. Root $T$ at $e_1$. Now, all subtrees pending from $e_1$ are splitting (the edge in $\overline \partial T$ forms its own subtree, so it is splitting) and have at most $\lceil\frac m 2\rceil$ edges, and we can apply Lemma~\ref{lem:blockDyn}.
	
	\item \label{split:case3} If $|\overline \partial T| = 2$, with $e$ and $f$ the two edges on the internal boundary, and $v$ is on the path between $e$ and~$f$. Without loss of generality, assume that $T_1$ and $T_2$ contain $e$ and $f$ respectively.
	\begin{itemize}
		\item[-] If exactly one of $T_1$ or $T_2$ is not splitting, w.l.o.g. we can assume that it is $T_1$. By rooting $T$ at $e_1$, all subtrees pending from $e_1$ are splitting and contain at most $\lceil\frac m 2\rceil$ edges, and we can apply Lemma~\ref{lem:blockDyn}.
		\item[-] If both $T_1$ and $T_2$ are not splitting, write $e_1 = (v_1, v)$. Since $T_1$ is not splitting, $v_1$ must be incident to $e$. Moreover, since $T_2$ is not splitting, $e_1$ and $f$ are not incident, and all the subtrees pending at $v_1$ are splitting. We apply Lemma~\ref{lem:blockDyn} a first time by rooting $T$ at $v_1$. Let $T'$ be the subtree hanging from $v_1$ that contains $v$ and root $T'$ at $e_2$. Then, all subtrees pending from $e_2$ are splitting, and we can apply Lemma~\ref{lem:blockDyn} a second time. The resulting subtrees all have at most $\lceil\frac m 2\rceil$ edges.
	\end{itemize}
	
	\item \label{split:case4} Finally, if $v$ is not on the path between $e$ and $f$, let $v'$ be the vertex on this path which is closest to $v$. We can first split at $v'$ by applying the procedure from the case~\ref{split:case3}. All the resulting subtrees have at most $\lceil\frac m 2\rceil$ edges except maybe the subtree $T'$ containing $v$. However $T'$ has a fringe boundary of size $1$, and by splitting one more time according to either case~\ref{split:case1} or case~\ref{split:case2}, all resulting subtrees are splitting, and have at most $\lceil\frac m 2\rceil$ edges. In a worst case scenario, we needed to use Lemma~\ref{lem:blockDyn} three times.
\end{enumerate}
Let $T'_1, \ldots, T'_s$ be the subtrees into which $T$ is split by applying the procedure above, and let $E_i'$ be the set of edges of $T_i'$. Since $T_i'$ is splitting, by the induction hypothesis, $\cL_{E_i'}^\sigma$ is ergodic for every $\sigma\in \Omega^\mu_{E}$, so $\cL_{E}^\mu$ is also ergodic. Recall that $\tau_i:= \max_{\sigma \in \Omega_{E}^\mu} \tau (\cL_{E_i'}^\sigma)$.
Lemma~\ref{lem:blockDyn} shows that there exists $K$ such that if $T$ is split at a vertex/edge into $r$ subtrees $T_1',\dots, T_r'$, then $\tau(\cL_{E}^\mu)\leq K (d^3+ \sum_{i=1}^r \tau_i)$.
As we use Lemma~\ref{lem:blockDyn} at most three times in each step of the splitting procedure described above, we have
\begin{align*}
	\tau(\cL_{E}^\mu) \leq K^3\Big(3d^3 + \sum_{i=1}^{s} \tau_i \Big) \; ,
\end{align*}
Using the induction hypothesis on $T'_i$, if we denote by $m_i \leq \lceil\frac m 2\rceil$ the number of edges in $T'_i$, we have:
\begin{align*}
	\tau(\cL_{E}^\mu) &\leq K^3\Big(3d^3 + \sum_{i=1}^{s} d^3 m_i^C\Big)  
	\leq K^3 \Big(3d^3 + d^3 \left\lceil\frac m 2\right\rceil^{C-1} \sum_{i=1}^{s} m_i \Big) \\
	&\leq K^3 d^3 \left(3 + \frac{m^C}{2^{C-2}}\right) 
	\leq d^3  m^C \cdot \frac {K^3} {2^{C-3}}\;.
\end{align*}
So letting $C \geq 3(1+ \log{K})$ gives the desired inequality, and we conclude the proof of Claim~\ref{cla:1}.

\section{Further comments and open problems}
Our main result, Theorem~\ref{thm:stronger} gives a polynomial upper bound on the mixing time for the Glauber dynamics on edge-colourings of a tree. We made no effort to optimise the constant $C$, but it can be bounded by $60$.

In the case of complete trees of degree $\Delta$, we can obtain a better exponent. Using Lemma~\ref{lem:rel_time_clique} and the canonical decomposition of the complete tree of height $h$ into $\Delta$ complete subtrees of height $h-1$, one can directly obtain an upper bound for the relaxation time of the discrete-time dynamics of the form $n^{2 + o_d(1)}$, which by using~\eqref{eq:cont_disc} gives a $n^{3 + o_d(1)}$ bound for the mixing time. One can prove an analogue of Lemma 26 in~\cite{TVVY} (or Theorem 5.7 in~\cite{MSV04}) for edge-colourings of complete trees, which improves the bound on	 the mixing time to $n^{2+o_d(1)}$.

If the complete tree is rooted at $r$ and $k\in [\Delta+1, 2\Delta-1]$, one can construct a $k$-edge-colouring such that the ball of radius $R=\lfloor\log_{\Delta}{n}\rfloor$ around $r$ is frozen. Using Lemma 4.2 in~\cite{HaySin} for the discrete-time dynamics, this shows that the mixing time of the Glauber dynamics for $k$-edge-colourings is $\Omega(n\log{n})$. An interesting open problem would be to determine if the mixing time of the Glauber dynamics for $k$-edge-colourings of complete trees is $O(n\log{n})$, or to obtain the weaker bound $n^{1+o_d(1)}$ (see e.g.~\cite{MSW,SlyZha,TVVY}). The main bottleneck in our proof is the bound on the relaxation time of the constrained clique given in Lemma~\ref{lem:rel_time_clique}. For $k=d+1$, $t=1$ and $p_i=1$, one can show that its relaxation time is $\Theta(d)$, so new ideas are required.

Our theorem applies to $\Delta\geq 3$ and $k\geq \Delta+1$. For the case $\Delta=2$,  the graph is a path, so vertex-colourings and edge-colourings are essentially the same. Dyer, Goldberg and Jerrum~\cite{DGJ} proved that Glauber dynamics for $k=3$ mixes in time $\Theta(n^3\log{n})$.

More generally, there are very few results on Glauber dynamics specific to edge-colourings of graphs. In particular, it would be interesting to find ways to obtain a substantial improvement over the bound $\frac{11\Delta}{3}$ obtained from~\cite{Vig} on the number of colours needed for polynomial mixing of the edge-colouring Glauber dynamics on any graph with maximum degree $\Delta$.

Part of our motivation for analysing edge-colourings of trees comes from the study of Glauber dynamics in sparse random graphs. It is an interesting open problem to determine for which $k$ the Glauber dynamics for $k$-edge-colourings of a random $\Delta$-regular graph is ergodic, and for which $k$ it mixes in polynomial time.

Finally, the problem of sampling uniformly random edge-colourings of some simple graphs is closely related to some fundamental questions in combinatorics. For instance, $2n$-edge-colourings of the complete graph $K_{2n}$ correspond to $1$-factorisations, and $n$-edge-colourings of the complete bipartite graph $K_{n,n}$ are in bijection with Latin squares. Markov Chain Monte Carlo methods have been introduced to sample such combinatorial objects~(see e.g.~\cite{DotLin,JacMat}), but it is not known if they rapidly mix. A promising approach to the problem is to find the smallest $k=k(n)$ for which Glauber dynamics mixes in polynomial time for $k$-edge-colourings of $K_{2n}$ or $K_{n,n}$.

	\medskip

\noindent \textbf{Acknowledgments.} We thank the referees for carefully reading the article, for their suggestions and for bringing the reference~\cite{SlyZha} to our attention.

\bibliographystyle{plain}
\bibliography{trees}	
	
\newpage

\appendix

\section{Missing proofs from Section~\ref{sec:comparison}}

	For any function $f: \Omega \to \mathbb{R}$, the \emph{variance} and the \emph{Dirichlet form} of $\cL$ are defined respectively as
	\begin{align*}
	\text{Var}_{\cL}(f) &= \frac{1}{2} \sum_{\alpha, \beta \in \Omega} \pi(\alpha) \pi(\beta) (f(\alpha) - f(\beta))^2\;,\\
	\xi_{\cL}(f, f)& = \frac{1}{2} \sum_{\alpha, \beta \in \Omega} \pi(\alpha) \cL[\alpha \rightarrow \beta] (f(\alpha) - f(\beta))^2 \;.
	\end{align*}
	Let $\cF_{\cL}=\{f :\, \text{Var}_{\cL}(f)> 0\}$ and note that $\cF_{\cL}=\cF_{\cL'}$ are the set of non-constant functions, as $\pi$ and $\pi'$  are positive on $\Omega$.
	It is well-known (see e.g. Remark 13.13 in~\cite{LevPer}) that the spectral gap of $\cL$ satisfies 
	\begin{align}\label{eq:rel_time_char}
	\Gap(\cL) = \min_{f\in \cF_\cL} \frac {\xi_{\cL}(f,f)}{\text{Var}_{\cL}(f)}\;.
	\end{align}

\subsection{Proof of Propostion~\ref{prop:weighted_frac_paths}}
	
	We will compare the Dirichlet form and the variance of $\cL$ and $\cL'$. We have
	\begin{align*}
	\xi_{\cL'}(f, f) &= \frac 1 2 \sum_{\alpha, \beta \in \Omega} \pi'(\alpha) \cL'[\alpha \rightarrow \beta] (f(\alpha) - f(\beta))^2 \\
	&= \frac 1 2 \sum_{\alpha, \beta \in \Omega} \sum_{\gamma \in \Gamma_{\alpha, \beta}} g_{\alpha,\beta}(\gamma) \pi'(\alpha) \cL'[\alpha \rightarrow \beta] (f(\alpha) - f(\beta))^2 \\
	&= \frac 1 2 \sum_{\alpha, \beta \in \Omega} \sum_{\gamma \in \Gamma_{\alpha, \beta}} g_{\alpha,\beta}(\gamma) \pi'(\alpha) \cL'[\alpha \rightarrow \beta] \left(\sum_{(\sigma, \eta) \in \gamma} \sqrt{\omega(\sigma, \eta)} \frac{f(\sigma) - f(\eta)}{\sqrt{\omega(\sigma, \eta)}}\right)^2 \\
	&\leq \frac 1 2 \sum_{\alpha, \beta \in \Omega} \sum_{\gamma \in \Gamma_{\alpha, \beta}} g_{\alpha,\beta}(\gamma) \pi'(\alpha) \cL'[\alpha \rightarrow \beta] \cdot |\gamma|_\omega \sum_{(\sigma, \eta) \in \gamma} \frac{(f(\sigma) - f(\eta))^2}{\omega(\sigma, \eta)} \\
	&= \frac 1 2 \sum_{(\sigma, \eta)\in \cL}(f(\sigma) - f(\eta))^2\cdot \frac{1}{\omega(\sigma, \eta)} \sum_{(\alpha,\beta)\in \cL'}\sum_{\gamma\in\Gamma_{\alpha, \beta}\atop \gamma\ni  (\sigma, \eta)}g_{\alpha,\beta}(\gamma) |\gamma|_\omega \pi'(\alpha)\cL'[\alpha \rightarrow \beta] \\
	&\leq \rho_{\max} \,\xi_{\cL}(f, f) \;,
	\end{align*}
	where we used the Cauchy-Schwartz inequality in the first inequality. Additionally, we have
	\begin{align*}
	\text{Var}_{\cL'}(f) &= \frac{1}{2}\sum_{\alpha, \beta \in \Omega}\pi'(\alpha) \pi'(\beta) (f(\alpha) - f(\beta))^2 \\
	&\geq \frac{1}{2b^2} \sum_{\alpha, \beta \in \Omega}\pi(\alpha) \pi(\beta) (f(\alpha) - f(\beta))^2 \\
	&= \frac{1}{b^2} \text{Var}_{\cL}(f)\;.
	\end{align*}
	Combining the previous two inequalities and using~\eqref{eq:rel_time_char}, the desired result follows,
	$$ 
	\Gap(\cL') = \min_{f\in \cF_{\cL'}} \frac {\xi_{\cL'}(f, f)}{\text{Var}_{\cL'}(f)} \leq b^2 \rho_{\max}\, \min_{f\in \cF_{\cL}} \frac {\xi_{\cL}(f, f)}{\text{Var}_{\cL}(f)} = b^2\rho_{\max}\, \Gap(\cL)\;. 
	$$

\subsection{Proof of Proposition~\ref{prop:mono}}

Let $U=\{u_1,\dots,u_m\} = V\setminus \{v\}$. Let $p_1,\dots, p_m,p_v$ be the parameters for $u_1,\dots,u_m,v$, respectively. Let $\Omega_V$ and $\Omega_{U}$ be the set of $L$-colourings of $G$ and $G[U]$ respectively. Let $d = |N(v)|$ the degree of $v$. Since $N(v)$ is a clique, we have: 
	\begin{align*}
	|\Omega_V| = |\Omega_U| (k-d) \;.
	\end{align*}
	Indeed, for every colouring $\alpha_U$ of $U$, there are exactly $(k-d)$ possibilities to extend it into a colouring of $V$.
	Let $\pi_V$ and $\pi_U$ be the stationary distributions of $\cL_V$ and $\cL_U$, which are uniform as the transitions are symmetric. Recall that $\cF_{V}$ is the set of non-constant functions from $\Omega_V$ to $\mathbb{R}$. Let $\cF_{V}^v$ be the subset of these functions which are independent of $v$, i.e. which satisfy $f(\alpha) = f(\beta)$ whenever $\alpha$ and $\beta$ agree on $U$.
	
Using~\eqref{eq:rel_time_char}, we have
	\begin{align}
	\label{eq:minf}
	\Gap(\cL_V) = \min_{f \in \cF_{V}} \frac {\cE_{\cL_V}(f)} {\Var_{\cL_V} (f)} \leq \min_{f \in \cF_{V}^v} \frac {\cE_{\cL_V}(f)} {\Var_{\cL_V} (f)} \;.
	\end{align}
	Fix $f\in \cF_V^v$. Then we have 
	\begin{align*}
	\Var_{\cL_V}(f) &= \frac 1 2 \sum_{\alpha, \beta \in \Omega_V} \pi_V(\alpha) \pi_V(\beta) (f(\alpha) - f(\beta))^2 \\ 
	&= \frac 1 2 \sum_{\alpha, \beta \in \Omega_V} \frac 1 {|\Omega_V|^2} (f(\alpha) - f(\beta))^2 \\ 
	&= \frac 1 2 \sum_{\alpha_U, \beta_U \in \Omega_U} \sum_{c \in [k] \atop c \not \in \alpha(N(v))} \sum_{c' \in [k] \atop c' \not \in \beta(N(v))}  \frac 1 {|\Omega_V|^2} (f(\alpha_U) - f(\beta_U))^2 \\
	&= \frac 1 2 \sum_{\alpha_U, \beta_U \in \Omega_U} \frac {(k-d)^2} {|\Omega_V|^2} (f(\alpha_U) - f(\beta_U))^2  \\ \stepcounter{equation}\tag{\theequation}\label{eq:Var}
	&=  \Var_{\cL_U}(f)
	\end{align*}
	and
	\begin{align*}
	\cE_{\cL_V}(f)  
	&= \frac 1 2 \sum_{u_i \in U}\sum_{\alpha, \beta \in \Omega_V \atop \text{differ at } u_i} \frac 1 {|\Omega_V|} p_i (f(\alpha) - f(\beta))^2 +  \sum_{\alpha, \beta \in \Omega_V \atop \text{differ at } v} \frac 1 {|\Omega_V|} p_v (f(\alpha) - f(\beta))^2\\
	&= \frac 1 2 \sum_{u_i \in U}\sum_{\alpha_U, \beta_U \in \Omega_U \atop \text{differ at } u_i} \sum_{c \in [k] \atop {c \not \in \alpha(N(v)) \atop c \not \in \beta(N(v))}} \frac 1 {|\Omega_V|} p_i (f(\alpha_U) - f(\beta_U))^2\\
	&\leq \frac 1 2 \sum_{u_i \in U}\sum_{\alpha_U, \beta_U \in \Omega_U \atop \text{differ at } u_i} \frac {k-d} {|\Omega_V|} p_i (f(\alpha_U) - f(\beta_U))^2 \\
	&= \cE_{\cL_U}(f) \;. \stepcounter{equation}\tag{\theequation}\label{eq:cE}
	\end{align*}
	where we used that $N(v)$ is a clique in the inequality.
	Putting together~\eqref{eq:minf}--\eqref{eq:cE}, gives as required
	\begin{align*}
	\Gap(\cL_V) \leq \Gap(\cL_U) \;.
	\end{align*}

\end{document}